\documentclass[12pt,onecolumn,draftcls]{IEEEtran}
\usepackage{amsfonts}
\usepackage{graphicx}
\usepackage{color}
\usepackage{amsmath,amsfonts,amssymb,amsthm,epsfig,epstopdf,url,array}
\usepackage{url,textcomp}
\usepackage{authblk}
\usepackage{cite}
\makeatletter
\newcommand*{\rom}[1]{\expandafter\@slowromancap\romannumeral #1@}
\makeatother
\newcommand{\bs}{\boldsymbol}
\newtheorem{theorem}{Theorem}
\newtheorem{lemma}{Lemma}

\newtheorem{property}{Property}

\begin{document}
\title{Asymptotic Outage Analysis of HARQ-IR over Time-Correlated Nakagami-$m$ Fading Channels}
\author{Zheng~Shi,
        Shaodan~Ma, Guanghua~Yang, Kam-Weng~Tam, and Minghua Xia
\thanks{Zheng Shi, Shaodan Ma and Kam-Weng Tam are with the Department of Electrical and Computer Engineering, University of Macau, Macao (e-mail:shizheng0124@gmail.com, shaodanma@umac.mo, kentam@umac.mo).}
\thanks{Guanghua Yang is with the Institute of Physical Internet, Jinan University, China (e-mail: ghyang@jnu.edu.cn).}
\thanks{Minghua Xia is with the Department of Electronics and Communication Engineering, Sun Yat-Sen University, Guangzhou, China (email: xiamingh@mail.sysu.edu.cn).}
}
\maketitle
\begin{abstract}
In this paper, outage performance of hybrid automatic repeat request with incremental redundancy (HARQ-IR) is analyzed. Unlike prior analyses, time-correlated Nakagami-$m$ fading channel is considered. The outage analysis thus involves the probability distribution analysis of a product of multiple correlated shifted Gamma random variables and is more challenging than prior analyses. Based on the finding of the conditional independence of the received signal-to-noise ratios (SNRs), the outage probability is exactly derived by using conditional Mellin transform. Specifically, the outage probability of HARQ-IR under time-correlated Nakagami-$m$ fading channels can be written as a weighted sum of outage probabilities of HARQ-IR over independent Nakagami fading channels, where the weightings are determined by a negative multinomial distribution. This result enables not only an efficient truncation approximation of the outage probability with uniform convergence but also asymptotic outage analysis to further extract clear insights which have never been discovered for HARQ-IR even under fast fading channels. The asymptotic outage probability is then derived in a simple form which clearly quantifies the impacts of transmit powers, channel time correlation and information transmission rate. It is proved that the asymptotic outage probability is an inverse power function of the product of transmission powers in all HARQ rounds, an increasing function of the channel time correlation coefficients, and a monotonically increasing and convex function of information transmission rate. The simple expression of the asymptotic result enables optimal power allocation and optimal rate selection of HARQ-IR with low complexity. Finally, numerical results are provided to verify our analytical results and justify the application of the asymptotic result for optimal system design.

\end{abstract}
\IEEEpeerreviewmaketitle

\section{Introduction}\label{sec:int}
Recently, hybrid automatic repeat request (HARQ) technique has found wide applications in wireless communications due to its high potential for reliable transmissions. It has been proved from an information-theoretic view in \cite{caire2001throughput} that HARQ with incremental redundancy (HARQ-IR) can achieve the ergodic capacity in Gaussian collision channels. It also provides superior performance over other types of HARQ since extra coding gain is obtained through code combining. Thus this paper focuses on the analysis of HARQ-IR. As shown in \cite{makki2014performance}, the most fundamental metric to evaluate the performance of HARQ-IR is outage probability and its analysis essentially turns to determine the cumulative distribution function (CDF) of accumulated mutual information.

In prior literature, various methods have been proposed to derive the CDF of accumulated mutual information for HARQ-IR under either quasi-static \cite{shen2009average,makki2013green} or fast fading channels \cite{to2015power,kim2014optimal,sesia2004incremental,jinho2013energy,stanojev2009energy,
aghajanzadeh2011outage,yilmaz2010outage,chelli2013performance,
larsson2014throughput,larsson2016throughput}. To be more specific, for quasi-static fading channels where the same channel realization is experienced by the transmitted signal in each HARQ round, the CDF of accumulated mutual information is easy to be derived because of the simplicity of handling only one single random variable (RV). Hereby, \cite{shen2009average} and \cite{makki2013green} have conducted average rate analysis and power optimization for HARQ-IR, respectively. Unfortunately, the analytical results under quasi-static fading channels are only applicable to low mobility environment. In high mobility environment, the transmitted signals among all HARQ rounds would experience fast fading channels, where the channel responses vary independently from one transmission to another. Under fast fading channels, the derivation of the CDF of accumulated mutual information essentially turns to determine the distribution of the product of multiple independent shifted RVs. There are several approaches proposed to tackle this problem in the literature. For example, in \cite{to2015power,kim2014optimal,sesia2004incremental}, Log-normal approximation is proposed based on central limit theorem (CLT). In \cite{jinho2013energy,stanojev2009energy,aghajanzadeh2011outage}, a lower bound and an upper bound of the CDF are derived based on Jensen's inequality and Minkowski inequality, respectively. To calculate the exact CDF, Mellin transform and multi-fold convolution have been applied in \cite{yilmaz2010outage,chelli2013performance} and \cite{larsson2014throughput} respectively. Specifically, based on Mellin transform,
\cite{yilmaz2010outage} derives the exact CDF in terms of generalized Fox's H function. The analytical results are then applied for outage analysis in \cite{chelli2013performance}. In \cite{larsson2014throughput}, noticing that the accumulated mutual information is a sum of independent RVs under fast fading channels, the exact CDF is thus derived by using multi-fold convolution. Moreover, by associating the analytical results in \cite{larsson2014throughput} with matrix exponential distribution through Laplace transform, \cite{larsson2016throughput} approximates the CDF of accumulated mutual information to a matrix exponential distribution. Unfortunately, the exact results in \cite{yilmaz2010outage,chelli2013performance,larsson2014throughput}
are too complicated to provide meaningful insights due to difficulties in handling generalized Fox's H function and multi-fold integrals, while the approximated results in \cite{to2015power,kim2014optimal,sesia2004incremental} and \cite{larsson2016throughput} involve certain numerical calculations which also hinder the extraction of clear insights.


The analysis of HARQ-IR becomes more challenging while considering time-correlated fading channels, which usually occur in low-to-medium mobility environment, because a product of multiple shifted correlated RVs is involved in the CDF of accumulated mutual information. To the best of our knowledge, there are only few approximation approaches available to analyze the performance of HARQ-IR over time-correlated fading channels, that is, Log-normal approximation \cite{yang2014performance}, polynomial fitting technique \cite{shi2015analysis} and inverse moment matching method \cite{shi2016inverse}. Unfortunately, the Log-normal approximation in \cite{yang2014performance} is inaccurate when the fading channels have medium-to-high time correlation. Although the results in \cite{shi2015analysis} and \cite{shi2016inverse} usually can achieve a good approximation of the CDF of accumulated mutual information, their tightness is proved in mean square error (MSE) sense. It does not necessarily imply that the approximation error always approaches to zero and the occurrence of inaccuracy is found especially under low outage probability/hign signal-to-noise ratio (SNR). In addition, the approximated results in \cite{shi2015analysis,shi2016inverse} are still complicated with little insights, thus limiting their applications in practical system design.

In this paper, asymptotic outage analysis is conducted to thoroughly investigate the performance of HARQ-IR over time-correlated Nakagami-$m$ fading channels with meaningful insights. Based on our finding that the received SNRs in multiple HARQ transmissions are conditionally independent given a certain RV, the CDF of accumulated mutual information can be exactly derived by using conditional Mellin transform. With the result, the outage probability can be derived as a weighted sum of outage probabilities of HARQ-IR over independent Nakagami fading channels, where the weightings are determined by a negative multinomial distribution. A truncation approximation with uniform convergence is then proposed to ease the computation of the exact outage probability. Notice that the uniform convergence is stronger than the MSE convergence in \cite{shi2015analysis,shi2016inverse} and thus offers higher accuracy in the approximation. The analytical result in this paper also enables the asymptotic outage analysis, which has never been conducted for HARQ-IR even under fast fading channels. More specifically, the asymptotic outage probability is derived in a simple form which clearly quantifies the impacts of transmission powers in multiple transmissions, channel time correlation and information transmission rate. It is also proved that the asymptotic outage probability is an inverse power function of the product of transmission powers in all HARQ rounds, an increasing function of the channel time correlation coefficients, and a monotonically increasing and convex function of information transmission rate. The simple expression of the asymptotic result enables optimal power allocation and optimal rate selection of HARQ-IR through plenty of readily available optimization techniques with low complexity. Finally, numerical results are provided to verify our analytical results and justify the application of the asymptotic result for optimal system design.




The remainder of this paper is organized as follows. Section \ref{sec:sys_mod} introduces the system model and formulates the outage probability. The probability distribution of the product of multiple correlated shifted RVs which is necessary for outage analysis is derived and the exact outage probability is given in Section \ref{sec:exa}. Asymptotic outage probability is then derived in a simple form and meaningful insights are shown in Section \ref{sec:asy}. In Section \ref{sec:num_res}, numerical results are presented for validation and discussion. Finally, Section \ref{sec:con} concludes this paper.

\section{System Model and Outage Formulation}\label{sec:sys_mod}
A point-to-point HARQ-IR enabled system operating over time-correlated Nakagami-$m$ fading channels is considered in this paper. The details of the system are introduced as follows.


\subsection{HARQ-IR Protocol}\label{sec:harq_ir}
Following HARQ-IR protocol, the source first encodes every $b$-bits information message into $K$ codewords each with length of $L$, where $K$ denotes the maximum allowable number of transmissions for each message. As a result, the initial coding rate is $c = b/L$ \cite{wu2010performance}. The $K$ codewords will be sequentially transmitted to the destination until the message is successfully decoded. In each transmission, the previously received codewords are combined with the most recently received codeword for joint decoding. If successful, an acknowledgement (ACK) message is fed back from the destination to the source and the source then moves to the transmission of the next information message. Otherwise, a negative acknowledgement (NACK) message is fed back from the destination to the source and the source transmits the next codeword until the maximum number of transmissions $K$ is reached. Similarly to the analyses in the literature \cite{chelli2013performance,larsson2014throughput,larsson2016throughput}, error-free feedback channels are assumed here, that is, all feedback messages can be successfully decoded.
%
%
\subsection{Time-Correlated Nakagami-$m$ Fading Channels}\label{sec:corr}
Denote ${\bf x}_k$ as the $k$th codeword with length of $L$. It is transmitted over a block fading channel, i.e., each symbol of ${\bf x}_k$ experiences an identical channel realization during the $k$th transmission. Therefore, the signal received in the $k$th transmission is given by
\begin{equation}\label{eqn:sign_mod}
  {\bf y}_k = h_k {\bf x}_k + {\bf n}_k,
\end{equation}
where ${\bf n}_k$ denotes a complex additive white Gaussian noise (AWGN) vector with zero mean vector and covariance matrix $\mathcal N_0{\bf I}_{L}$, i.e., ${\bf n}_k \sim {\cal CN}(0,\mathcal N_0{{\bf I}_{L}})$, ${\bf I}_{L}$ represents an $L \times L$ identity matrix, and $h_k$ denotes the block fading channel coefficient in the $k$th transmission.



Notice that Nakagami-$m$ fading is a general channel model whose parameters can be adjusted to fit a variety of empirical measurements \cite{goldsmith2005wireless}, e.g., it covers one-sided Gaussian channel and Rayleigh fading as special cases by setting $m=\frac{1}{2}$ and $m=1$ respectively. It is thus considered here. Unlike most of the prior analyses, time-correlated fading channels are considered, that is, the channel coefficients among the $K$ transmissions are correlated. The time correlation usually occurs in low-to-medium mobility environment \cite{kim2011optimal}. Herein, a widely used Nakagami-$m$ fading channel model with generalized correlation is adopted and the channel magnitude $|h_k|$ is written as \cite{beaulieu2011novel,lopez2013bivariate,shi2012correlated}
\begin{equation}\label{eqn:R_k_def}
|{h_k}| = \sqrt {\frac{{{\sigma _k}^2}}{m}} \left\| {\sqrt {1 - {\lambda _k}^2} \left( {\begin{array}{*{20}{c}}
{{\vartheta _{k,1}}}\\
 \vdots \\
{{\vartheta _{k,m}}}
\end{array}} \right) + {\lambda _k}\left( {\begin{array}{*{20}{c}}
{{\vartheta _{0,1}}}\\
 \vdots \\
{{\vartheta _{0,m}}}
\end{array}} \right)} \right\| ,
\end{equation}
where $\left\| \cdot \right\|$ represents ${\ell }^2$ norm, ${\vartheta_{k,l}} $ and ${\vartheta_{0,l}}$ are independent circularly-symmetric complex Gaussian random variables (RVs) with zero mean and unit variance, i.e., ${\vartheta_{0,l}}, {\vartheta_{k,l}} \sim \mathcal{CN}\left( {0,1} \right)$, $m$ represents the fading order that indicates the severity of fading channels, ${\sigma_k}^2$ denotes Nakagami spread and is equal to the expectation of the squared channel magnitude, i.e., ${\rm  E}\{|h_k|^2\}={\sigma_k}^2$, and ${\boldsymbol{\lambda }} = \left( {{\lambda _1},{\lambda _2}, \cdots ,{\lambda _K}} \right)$ specifies the generalized time correlation among fading channels of all HARQ transmissions \cite{beaulieu2011novel}. Under this model, the magnitude of channel coefficient ${\left| {{h_{k}}} \right|}$ follows a Nakagami-$m$ distribution, i.e., $\left| {{h_k}} \right| \sim {\rm Nakagami}(m,{\sigma_k}^2)$, and the channel magnitudes ${\left| {{h_{l}}} \right|}$ and ${\left| {{h_{k}}} \right|}$ are correlated. Moreover, the cross correlation coefficient between the squared channel magnitudes ${\left| {{h_{l}}} \right|^2}$ and ${\left| {{h_{k}}} \right|^2}$ is determined by the time correlation coefficients $\bs \lambda$ as \cite{beaulieu2011novel}
\begin{align}
\label{eqn_cor_fa}
 \frac{{{\rm{E}}\left( {{{\left| {{h_{l}}} \right|}^2}{{\left| {{h_{k}}} \right|}^2}} \right) - {\rm{E}}\left( {{{\left| {{h_{l}}} \right|}^2}} \right){\rm{E}}\left( {{{\left| {{h_{k}}} \right|}^2}} \right)}}{{\sqrt {{\rm{Var}}\left( {{{\left| {{h_{l}}} \right|}^2}} \right){\rm{Var}}\left( {{{\left| {{h_{k}}} \right|}^2}} \right)} }}
= {\lambda _{l}}^2{\lambda _{k}}^2,\, 1 \le l \ne k \le K,
\end{align}
where ${\rm Var}(\cdot)$ denotes the variance of the random variable in the brackets. Without loss of generality, the time correlation coefficients are assumed to be non-negative, i.e., $0 \le \lambda _1,\cdots,\lambda _K \le 1$. Clearly from (\ref{eqn:R_k_def}), when ${\boldsymbol{\lambda }} = {\bf 1}_K$, the time-correlated fading channel reduces to quasi-static fading channel with $h_1=h_2=\cdots=h_K$, while when ${\boldsymbol{\lambda }} = {\bf 0}_K$, the time-correlated fading channel reduces to a fast fading channel where the channel coefficients $h_1$, $h_2$, $\cdots$, $h_K$ are mutually independent. In other words, the Nakagami-$m$ fading channel with generalized time correlation in (\ref{eqn:R_k_def}) includes quasi-static fading channel and fast fading channel as its special cases. Here ${\bf 1}_K$ and ${\bf 0}_K$ denote an all-ones vector and a null vector, each with length $K$, respectively. Unless otherwise indicated, subscript $K$ is omitted in the sequel.

According to (\ref{eqn:sign_mod}), the received signal-to-noise ratio (SNR) in the $k$th transmission is given by
\begin{equation}\label{eqn:SNR}
{\gamma _k} = \frac{{P_k}{\left| {{h_k}} \right|^2}}{\mathcal N_0} ,
\end{equation}
where $P_k$ denotes the transmitted signal power in the $k$th transmission. Since $\left| {{h_k}} \right|$ is Nakagami-$m$ distributed, i.e., $\left| {{h_k}} \right| \sim {\rm Nakagami}(m,{\sigma_k}^2)$, it is readily proved that ${\gamma _k} $ complies with \emph{Gamma distribution}, i.e., ${\gamma _k} \sim \mathcal G\left( {m,\frac{{P_k}{\sigma _k}^2}{m\mathcal N_0}} \right)$. Due to the time correlation among the channel coefficients as given in (\ref{eqn_cor_fa}), the SNRs $\bs \gamma = (\gamma _1,\gamma _2,\cdots,\gamma _K)$ are correlated Gamma RVs.

\subsection{Outage Formulation}
Outage probability has been proved as the most fundamental performance metric of HARQ schemes \cite{caire2001throughput}. For HARQ-IR, the outage probability is directly determined by the CDF of accumulated mutual information. Specifically, assuming that information-theoretic capacity achieving channel coding is adopted for HARQ-IR, an outage event happens when the information cannot be successfully decoded after $K$ transmissions, i.e.,  the accumulated mutual information $I_K$ is below the information transmission rate $\mathcal R$. Notice that the information transmission rate $\mathcal R$ depends on the coding rate $c$, modulation scheme and symbol transmission rate. The outage probability after $K$ transmissions is thus written as
\begin{equation}\label{eqn:out_prob_def}
{p_{out,K}} = \Pr \left( {{I_K}  < \mathcal R} \right) = F_{{I_K}}(\mathcal R).
\end{equation}
where $F_{{I_K}}(\cdot)$ denotes the CDF of ${I_K}$. In HARQ-IR protocol, decoding at each transmission is performed based on the combined codewords received in the previously transmissions and the current transmission and the accumulated mutual information after $K$ transmissions is given by
\begin{equation}\label{eqn:def:accum_inf}
{I_K} = \sum\nolimits_{k = 1}^K {{{\log }_2}\left( {1 + {\gamma _k}} \right)}.
\end{equation}
Accordingly, the outage probability ${p_{out,K}}$ becomes
\begin{equation}\label{eqn:out_prob_rew}
{p_{out,K}} = \Pr \left( {G_K \triangleq \prod\nolimits_{k = 1}^K {\left( {1 + {\gamma _k}} \right)}  < {2^{\mathcal R}}} \right) = {F_{G_K}}\left( {{2^{\mathcal R}}} \right),
\end{equation}
where $F_{G_K}(\cdot)$ denotes the CDF of $G_K$. Noticing that $(\gamma _1,\gamma _2,\cdots,\gamma _K)$ are correlated Gamma RVs, the derivation of outage probability ${p_{out,K}}$ essentially turns to determining the CDF of the product of multiple correlated shifted-Gamma RVs, i.e., ${F_{G_K}}\left( {{x}} \right)$. Due to the presence of time correlation and Gamma distribution, the outage analysis  is more challenging than those for quasi-static \cite{shen2009average,makki2013green} or fast fading channels \cite{to2015power,kim2014optimal,sesia2004incremental,jinho2013energy,stanojev2009energy,aghajanzadeh2011outage,
yilmaz2009productshifted,yilmaz2010outage,chelli2013performance,larsson2014throughput,larsson2016throughput}.  

It is worth noting that the CDF of the product of multiple shifted and correlated RVs has numerous applications in wireless communications and the outage formulation in (\ref{eqn:out_prob_def}) is also applicable to parallel transmission systems and orthogonal frequency-division multiplexing (OFDM) systems \cite{bai2013outage,luo2005service}. Specifically, in OFDM systems, coded signals are transmitted through multiple subcarriers. Due to close frequency spacing of the subcarriers and limited Doppler spread, the fading channels on multiple subcarriers are usually correlated \cite{lu2002ldpc}. The outage analysis in OFDM systems is thus also reduced to the analysis of the CDF of the product of multiple shifted and correlated RVs in (\ref{eqn:out_prob_rew}).



\section{Probability Distribution of the Product of Multiple Correlated Shifted-Gamma RVs}\label{sec:exa}
As aforementioned, the derivation of outage probability ${p_{out,K}}$ essentially turns to determining the CDF of the product of multiple shifted and correlated Gamma RVs. It is mathematically difficult due to the presence of time correlation. In the literature, only a few \emph{approximated} results are available in \cite{yang2014performance,shi2015analysis,shi2016inverse}. They are complicated without clear insights and thus are not favorable for system design. Moreover, their accuracy can not be guaranteed under certain scenarios. In this paper, we will derive the \emph{exact} CDF of the product of multiple shifted and correlated Gamma RVs, based on which a truncation approximation with uniform convergence can be proposed and asymptotic outage probability can be derived in a simple form with clear insights.

\subsection{Exact CDF}
From (\ref{eqn:R_k_def}) and (\ref{eqn:SNR}), we can find that the SNRs $\{\gamma _1,\gamma _2,\cdots,\gamma _K\}$ are conditionally independent, given the complex Gaussian RVs ${\vartheta_{0,1}},\cdots,{\vartheta_{0,m}}$. It has been proved in \cite[Theorem 1.3.4]{muirhead2009aspects} that $\{\gamma _1,\gamma _2,\cdots,\gamma _K\}$ follow \emph{independent noncentral chi-squared distributions} with $2m$ degrees-of-freedom when conditioned on $T \triangleq \sum\nolimits_{l = 1}^m {{{\left| {\vartheta_{0,l}} \right|}^2}} $. The conditional probability density function (PDF) is given in the following lemma.
\begin{lemma}\cite[Theorem 1.3.4]{muirhead2009aspects}
The conditional PDF of the SNR $\gamma_k$ given $T =t$ is written as
\begin{align}\label{eqn:cond_pdf_r_k_rew}
{{f_{\left. {{\gamma _k}} \right|T}}\left( {\left. {{x_k}} \right|t} \right)} &{ = {{\left( {\frac{1}{{{\Omega _k}}}} \right)}^m}\frac{{{x_k}^{m - 1}}}{{\Gamma \left( m \right)}}{e^{ - \frac{{{u_k}{\lambda _k}^2t + {x_k}}}{{{\Omega _k}}}}}_0{F_1}\left( {;m;{{\left( {\frac{{\sqrt {{u_k}{\lambda _k}^2{x_k}t} }}{{{\Omega _k}}}} \right)}^2}} \right),}\,  {{\lambda _k}} \ne 1,
\end{align}
where ${u_k} = \frac{{P_k}{\sigma _k}^2}{m\mathcal N_0}$, $\Omega _k = \frac{{P_k}{\sigma _k}^2\left( {1 - {\lambda _k}^2} \right)}{m\mathcal N_0}$,  $\Gamma(\cdot)$ and ${}_0F_1(\cdot)$ denote Gamma function and the confluent hypergeometric limit function \cite[Eq.16.2.1]{olver2010nist}, respectively.
\end{lemma}
It should be mentioned that (\ref{eqn:cond_pdf_r_k_rew}) is not applicable to quasi-static fading channels, i.e., $\bs \lambda = \bf 1$, because of the occurrence of singularity. Due to its speciality, the outage analysis of HARQ-IR over quasi-static fading channels will be discussed separately. 

Noticing that Mellin transform\footnote{
The Mellin transform with respect to a function $f(x)$ is defined as $\left\{ {\mathcal M f} \right\}\left( s \right) = \int\nolimits_0^\infty  {{x^{s - 1}}f\left( x \right)dx} \triangleq \tilde f\left( s \right)$, and the associated inverse Mellin transform is given by $f\left( x \right) = \frac{1}{{2\pi {\rm{i}}}}\int\nolimits_{c - {\rm{i}}\infty }^{c + {\rm{i}}\infty } {{x^{ - s}}\tilde f\left( s \right)dx}$ \cite{debnath2010integral}.} is a powerful mathematical tool to obtain the probability distribution of a product of multiple independent RVs \cite{yilmaz2009productshifted,yilmaz2010outage,chelli2013performance}, the conditional independence of SNRs $\bs \gamma$ given $T$ inspires us to derive the conditional PDF of $G_K$ in closed-form by using conditional Mellin transform. Specifically, with the conditional independence, the Mellin transform of the conditional PDF of $G_K$ can be written as a product of the Mellin transforms corresponding to the conditional PDFs of the shifted SNRs. With this special property, as proved in Appendix \ref{app:phi_derivation}, the conditional PDF of $G_K$ given $T=t$, ${f_{\left. G_K \right|T}}(x|t)$, can be derived as
\begin{align}\label{eqn:g_hat_cond_t}
 {f_{\left. { G_K} \right|T}}\left( {\left. x \right|t} \right)
& = \sum\limits_{{\ell_1}, \cdots ,{\ell_K} = 0}^\infty  {{t^{\sum\nolimits_{k = 1}^K {{\ell_k}} }}{e^{ - t\sum\nolimits_{k = 1}^K {\frac{{{\lambda _k}^2}}{{1 - {\lambda _k}^2}}} }}\prod\limits_{k = 1}^K {\frac{1}{{{\ell_k}!}}{{{{\left( {\frac{{{\lambda _k}^2}}{{1 - {\lambda _k}^2}}} \right)}^{{\ell_k}}}}}} } {f_{{{ {\mathcal A}}_{\bs{\ell}}}}}(x),
\end{align}
where ${ \bs{\ell}}$ is a vector with $K$ elements as ${ \bs{\ell}} = (\ell_1,\cdots,\ell_K)$ and ${{f_{{{ {\cal A}}_{\bs{\ell}}}}}(x)}$ is given by
\begin{align}\label{eqn:pdf_produc_gamma_ind_sec}
{f_{{{ {\mathcal A}}_{\bs{\ell}}}}}(x) &= \frac{1}{{2\pi {\rm{i}}}}\int\nolimits_{c - {\rm{i}}\infty }^{c + {\rm{i}}\infty } {\prod\limits_{k = 1}^K {\frac{{\Psi \left( {m + {\ell_k},s + m + {\ell_k};\frac{1}{{{\Omega _k}}}} \right)}}{{{{\left( {{\Omega _k}} \right)}^{m + {\ell_k}}}}}} {x^{ - s}}ds}\notag\\
&= \frac{1}{{\prod\nolimits_{k = 1}^K {{\Omega _k}} }}Y_{0,K}^{K,0}\left[ {\left. {\begin{array}{*{20}{c}}
 - \\
{\left( {0,1,\frac{1}{{{\Omega _1}}},m + {\ell_1}} \right), \cdots ,\left( {0,1,\frac{1}{{{\Omega _K}}},m + {\ell_K}} \right)}
\end{array}} \right|\frac{x}{{\prod\nolimits_{k = 1}^K {{\Omega _k}} }}} \right],
\end{align}
with ${\rm i} = \sqrt{-1}$, $\Psi \left( {\alpha ,\gamma ;z} \right) = \frac{1}{{\Gamma \left( \alpha  \right)}}\int\nolimits_0^\infty  {{e^{ - zt}}{t^{\alpha  - 1}}{{\left( {1 + t} \right)}^{\gamma  - \alpha  - 1}}dt}$ denoting Tricomi's confluent hypergeometric function \cite[Eq.9.211.4]{gradshteyn1965table}, and $Y_{p,q}^{m,n}[\cdot]$ denoting the generalized Fox's H function defined by Mellin-Barnes integral as \cite{yilmaz2010outage,chelli2013performance}
\begin{equation}\label{eqn_fox_h_def}
Y_{p,q}^{m,n}\left[ {\left. {\begin{array}{*{20}{c}}
{\left( {{a_1},{\alpha _1},{ A_1},{\varphi _1}} \right), \cdots ,\left( {{a_p},{\alpha _p},{ A_p},{\varphi _p}} \right)}\\
{\left( {{b_1},{\beta _1},{B_1},{\phi _1}} \right), \cdots ,\left( {{b_q},{\beta _q},{B_q},{\phi _q}} \right)}
\end{array}} \right|x} \right] = \frac{1}{{2\pi {\rm i}}}\int_{\mathcal L} {M_{p,q}^{m,n}\left[ s \right]{x^{ - s}}ds},
\end{equation}
where $\mathcal{L}$ is a Mellin-Barnes contour in the complex $s$-plane running from $c-{\rm i}\infty $ to $c+{\rm i}\infty$, $c \in \mathbb{R}$ and $M_{p,q}^{m,n}\left[ s \right]$ is written as
\begin{multline}\label{eqn:def_M}
M_{p,q}^{m,n}\left[ s \right] = \frac{{\prod\nolimits_{j = 1}^m {{B_j}^{{\phi _j} + {b_j} + {\beta _j}s - 1}\Psi \left( {{\phi _j},{\phi _j} + {b_j} + {\beta _j}s;{B_j}} \right)} }}{{\prod\nolimits_{i = n + 1}^p {{ A_i}^{{\varphi _i} + {a_i} + {\alpha _i}s - 1}\Psi \left( {{\varphi _i},{\varphi _i} + {a_i} + {\alpha _i}s;{ A_i}} \right)} }}\times \\
\frac{{\prod\nolimits_{i = 1}^n {{ A_i}^{{\varphi _i} - {a_i} - {\alpha _i}s}\Psi \left( {{\varphi _i},{\varphi _i} + 1 - {a_i} - {\alpha _i}s;{ A_i}} \right)} }}{{\prod\nolimits_{j = m + 1}^q {{B_j}^{{\phi _j} - {b_j} - {\beta _j}s}\Psi \left( {{\phi _j},{\phi _j} + 1 - {b_j} - {\beta _j}s;{B_j}} \right)} }}.
\end{multline}
Notice that the generalized Fox's H function does have two special properties as shown in Appendix \ref{app:fox_H_fun} which can simplify our mathematical derivations, and an efficient MATHEMATICA{\textregistered} implementation of the generalized Fox's H function in (\ref{eqn_fox_h_def}) can be found in \cite{yilmaz2010outage}. Interestingly from (\ref{eqn:pdf_produc_gamma_ind_sec}), $f_{{\mathcal A}_{\bs{\ell}}}(x)$ is the inverse Mellin transform of a product function $\prod\nolimits_{k = 1}^K {\frac{{\Psi \left( {m + {\ell_k},s + m + {\ell_k};\frac{1}{{{\Omega _k}}}} \right)}}{{{{\left( {{\Omega _k}} \right)}^{m + {\ell_k}}}}}}$.
Here the multiplier ${\frac{{\Psi \left( {m + {\ell_k},s + m + {\ell_k};\frac{1}{{{\Omega _k}}}} \right)}}{{{{\left( {{\Omega _k}} \right)}^{m + {\ell_k}}}}}}$ in fact is the Mellin transform of the PDF of a shifted-Gamma RV $(1+R_{{\bs{\ell}},k})$, where $R_{{\bs{\ell}},k} \sim \mathcal {G}(m+\ell_k,{\Omega _k})$. Therefore, ${f_{{{ {\mathcal A}}_{\bs{\ell}}}}}(x)$ can be regarded as the PDF of the product of $K$ independent shifted-Gamma RVs i.e., $\mathcal A_{\bs{\ell}} = \prod\nolimits_{k = 1}^K (1+R_{{\bs{\ell}},k})$.
%

With (\ref{eqn:g_hat_cond_t}), the PDF of $G_K$ can be eventually obtained by integrating the conditional probability ${f_{\left. G_K \right|T}}\left( {\left. x \right|t} \right)$ over the distribution of $T$, such that
\begin{equation}\label{eqn_mar_G}
{{f_{G_K}}\left( x \right) = {\rm E}_{T}\left\{{f_{\left. G_K \right|T}}\left( x|t \right)\right\} = \int\nolimits_0^\infty  {{f_{\left. G_K \right|T}}\left( {\left. x \right|t} \right){f_T}\left( t \right)dt} },
\end{equation}
where ${f_T}\left( t \right)$ denotes the PDF of $T$.
With respect to the distribution of $T \triangleq \sum\nolimits_{l = 1}^m {{{\left| {\vartheta_{0,l}} \right|}^2}} $, since $\vartheta_{0,1},\cdots,\vartheta_{0,m}$ are independent and identically distributed (i.i.d.) complex Gaussian RVs with zero mean and unit variance, it can be proved that $\left|\vartheta_{0,1}\right|^2,\cdots,\left|\vartheta_{0,m}\right|^2$ are i.i.d. Gamma RVs, i.e., ${{\left| {\vartheta_{0,l}} \right|}^2} \sim \mathcal{G}\left({1,1}\right)$. As a sum of $m$ independent Gamma RVs with identical scale parameter, the random variable $T$ then follows a Gamma distribution with the PDF of \cite[p289]{degroot1986probability}
\begin{equation}\label{eqn_pdf_T_nak}
{f_T}\left( t \right) = \frac{{{t^{m - 1}}}}{{\Gamma \left( m \right)}}{e^{ - t}},\, t \ge 0.
\end{equation}

Plugging (\ref{eqn:g_hat_cond_t}) and (\ref{eqn_pdf_T_nak}) into (\ref{eqn_mar_G}), it follows by using \cite[Eq.3.381.4]{gradshteyn1965table} that
\begin{align}\label{eqn_pdf_f_g_1_p}
{f_{ G_K}}\left( x \right)
 &= \frac{1}{{\Gamma \left( m \right)}}\sum\limits_{{\ell_1}, \cdots ,{\ell_K} = 0}^\infty  {\int\nolimits_0^\infty  {{t^{\sum\limits_{k = 1}^K {{\ell_k}}  + m - 1}}{e^{ - t\left( {1 + \sum\limits_{k = 1}^K {\frac{{{\lambda _k}^2}}{{1 - {\lambda _k}^2}}} } \right)}}dt} \prod\limits_{k = 1}^K {\frac{1}{{{\ell_k}!}}{{{{\left( {\frac{{{\lambda _k}^2}}{{1 - {\lambda _k}^2}}} \right)}^{{\ell_k}}}}}{f_{{{ {\mathcal A}}_{\bs{\ell}}}}}(x)} } \notag\\
 &= \sum\nolimits_{{\ell_1}, \cdots ,{\ell_K} = 0}^\infty  {{W_{\bs{\ell}}}{f_{{{ {\mathcal A}}_{\bs{\ell}}}}}(x)},
\end{align}
where ${{\bs{\ell}}}$ is a vector with $K$ elements as ${{\bs{\ell}}} = (\ell_1,\cdots,\ell_K)$ and the weighting ${W_{\bs{\ell}}}$ is given as
\begin{equation}\label{eqn:W_l_def_sec}
{W_{\bs{\ell}}} \triangleq \frac{{\Gamma \left( {m + \sum\nolimits_{k = 1}^K {{\ell_k}} } \right)}}{{\Gamma \left( m \right)}}{\left( {1 + \sum\limits_{k = 1}^K {\frac{{{\lambda _k}^2}}{{1 - {\lambda _k}^2}}} } \right)^{ - m}}\prod\limits_{k = 1}^K {\frac{{{w_k}^{{\ell_k}}}}{{{\ell_k}!}}},\, {{\bs{\ell}=(\ell_1,\cdots,\ell_K)} \in {{\mathbb N}_0}^K},
\end{equation}
${w_k} = \frac{{{\lambda _k}^2}}{{1 - {\lambda _k}^2}}{\left( {1 + \sum\nolimits_{l = 1}^K {\frac{{{\lambda _l}^2}}{{1 - {\lambda _l}^2}}} } \right)^{ - 1}}$ and ${{\mathbb N}_0}^K$ denotes $K$-ary Cartesian power of natural number set. 
By comparing (\ref{eqn:W_l_def_sec}) with \cite[Eq.13.8-1]{bishop2007discrete}, it is found that $\left\{{W_{\bs{\ell}}},{{{\bs{\ell}}} \in {{\mathbb N}_0}^K}\right\}$  are probabilities of a negative multinomial distributed vector RV ${\bs{\ell}}$, i.e., ${\bs{\ell}} \sim {\rm NM}(m,\bf w)$, where ${\bf w}=(w_1,\cdots,w_K)$. Clearly, we have $\sum\nolimits_{{\bs{\ell}} \in {{\mathbb N}_0}^K}  {{W_{\bs{\ell}}}}  = 1$. As such, ${f_{ G_K}}\left( x \right)$ is expressed as a weighted sum of the PDFs corresponding to $\left\{\mathcal A_{\bs{\ell}},{{{\bs{\ell}}} \in {{\mathbb N}_0}^K}\right\}$.

Based on (\ref{eqn_pdf_f_g_1_p}), the CDF of the product of multiple correlated shifted-Gamma RVs ${F_{ G_K}}\left( x \right)$ can be derived as shown in the following theorem.
\begin{theorem}\label{the:cdf_pdf_corr_gam_shif}
The CDF of $G_K = \prod\nolimits_{k = 1}^K {\left( {1 + {\gamma _k}} \right)}$ is given by
\begin{equation}\label{eqn:CDF_G_def_sec}
{F_{ G_K}}\left( x \right) = \int\nolimits_0^x {{f_{ G_K}}\left( t \right)dt}  
=\sum\nolimits_{{\ell_1}, \cdots ,{\ell_K} = 0}^\infty  {{W_{\bs{\ell}}}\int\nolimits_0^x {{f_{{{ {\mathcal A}}_{\bs{\ell}}}}}(t)dt}}
= \sum\nolimits_{{\ell_1}, \cdots ,{\ell_K} = 0}^\infty  {{W_{\bs{\ell}}}{F_{{{ {\mathcal A}}_{\bs{\ell}}}}}(x)},
\end{equation}
where ${{F_{{{ {\mathcal A}}_{\bs{\ell}}}}}(x)}$ denotes the CDF of the product of independent shifted-Gamma RVs $\mathcal A_{\bs{\ell}} = \prod\nolimits_{k = 1}^K (1+R_{{\bs{\ell}},k})$ with $R_{{\bs{\ell}},k} \sim \mathcal {G}(m+\ell_k,{\Omega _k})$ and ${{\bs{\ell}}} = (\ell_1,\cdots,\ell_K)$. Specifically, ${{F_{{{ {\mathcal A}}_{\bs{\ell}}}}}(x)}$ can be expressed in terms of the generalized Fox's H function as
\begin{multline}\label{eqn:CDF_F_A_def_sec}
{F_{{{ {\mathcal A}}_{\bs{\ell}}}}}(x) = Y_{1,K + 1}^{K,1}\left[ {\left. {\begin{array}{*{20}{c}}
{\left( {1,1,0,1} \right)}  \\
{\left( {1,1,\frac{1}{{{\Omega _1}}},m + {\ell_1}} \right), \cdots ,\left( {1,1,\frac{1}{{{\Omega _K}}},m + {\ell_K}} \right),\left( {0,1,0,1} \right)}
\end{array}} \right|\frac{x}{{\prod\nolimits_{k = 1}^K {{\Omega _k}} }}} \right].
\end{multline}
\end{theorem}
\begin{proof}
Please see Appendix \ref{app:fox_H_fun}.

\end{proof}

Therefore, the outage probability ${p_{out,K}}$ can be obtained by substituting (\ref{eqn:CDF_G_def_sec}) into (\ref{eqn:out_prob_rew}), such that
\begin{equation}\label{eqn:out_prob_def_hat}
{p_{out,K}} = {F_{ G}}\left( {{2^{\mathcal R}} } \right) = \sum\nolimits_{{\ell_1}, \cdots ,{\ell_K} = 0}^\infty  {{W_{\bs{\ell}}}{F_{{{ {\mathcal A}}_{\bs{\ell}}}}}(2^{\mathcal R})}.
\end{equation}
Notice that $\mathcal A_{\bs{\ell}} = \prod\nolimits_{k = 1}^K (1+R_{{\bs{\ell}},k})$, $R_{{\bs{\ell}},k}$ follows Gamma distribution as $R_{{\bs{\ell}},k} \sim \mathcal {G}(m+\ell_k,{\Omega _k})$ and $\Omega _k = \frac{{P_k}{\sigma _k}^2\left( {1 - {\lambda _k}^2} \right)}{m\mathcal N_0}$. Clearly, the random variable $R_{{\bs{\ell}},k}$ can be factorized as $R_{{\bs{\ell}},k} = \frac{P_k\left| {{h_{{{\mathcal A}_{\bs{\ell}}},k}}} \right|^2}{\mathcal N_0}$ where $\left| {{h_{{{\mathcal A}_{\bs{\ell}}},k}}} \right|$ follows Nakagami distribution as
\begin{equation}\label{eqn:channel_coeff_dec}
\left| {{h_{{{\mathcal A}_{\bs{\ell}}},k}}} \right| \sim {\rm Nakagami}\left( {m + {\ell_k},\frac{\left( {m + {\ell_k}} \right)(1-{\lambda_k}^2){\sigma_k}^2}{m}} \right),\, 1 \le k \le K.
\end{equation}
Therefore, ${F_{{{ {\cal A}}_{\bs{\ell}}}}}\left( 2^{\mathcal R} \right)$ can be regarded as the outage probability of HARQ-IR after $K$ transmissions over independent Nakagami-$m$ fading channels, where $R_{{\bs{\ell}},k}$ denotes the received SNR in the $k$th HARQ transmission and ${{h_{{{\mathcal A}_{\bs{\ell}}},k}}}$ denotes the Nakagami-$m$ fading channel coefficient in the $k$th transmission. Consequently, the outage probability ${p_{out,K}}$ can be rephrased as a weighted sum of outage probabilities of HARQ-IR over independent Nakagami fading channels where the weightings ${W_{\bs{\ell}}}$ are probabilities of the negative multinomial distribution ${\rm NM}(m,\bf w)$.

It is worth mentioning that the outage result in (\ref{eqn:out_prob_def_hat}) is applicable to fast fading channels. Under fast fading channels, the time correlation coefficients are equal to zero,  i.e., $\bs \lambda =\bf 0$. Putting it into (\ref{eqn:W_l_def_sec}), we have $W_{\bf{0}} = 1$ and $W_{\bs{\ell}} = 0$ for all ${\bs{\ell}} \ne {\bf 0}$. Then the outage probability ${p_{out,K}}$ reduces to ${F_{{{\mathcal A}_{\bf 0}}}}({2^R})$.

\subsection{Truncation Approximation}
Considering that $p_{out,K}$ is represented by an infinite series in (\ref{eqn:out_prob_def_hat}), it is impossible to compute the exact value by adding an infinite number of terms up. To enable its computation, it is natural to truncate $p_{out,K}$ into a finite series. Towards this end, an effective truncation approach is proposed. Specifically, $p_{out,K}$ in (\ref{eqn:out_prob_def_hat}) is approximated with a truncation order $N$ as
\begin{equation}\label{eqn:F_G_x_trunc_ser}
 p_{out,K} \approx {\sum\nolimits_{{\sum\nolimits_{k=1}^K  \ell_k \le N}} {{W_{\bs{\ell}}}{F_{{ {\cal A}_{\bs{\ell}}}}}\left( 2^{\cal R} \right)} }  \triangleq \tilde p_{out,K}^N.
\end{equation}
%
%
The truncation error $\nabla(N)$ is characterized by the difference between $p_{out,K}$ and $\tilde p_{out,K}^N$, such that
\begin{equation}\label{eqn:trun_err_prod_gam}
 \nabla(N)= p_{out,K} - \tilde p_{out,K}^N = {\sum\nolimits_{{\sum\nolimits_{k=1}^K  \ell_k \ge N+1}}  {{W_{\bs{\ell}}}{F_{{ {\cal A}_{\bs{\ell}}}}}\left( 2^{\cal R} \right)} }  \ge 0.
\end{equation}
Clearly from (\ref{eqn:trun_err_prod_gam}), $\nabla(N)$ is a monotonically decreasing function of $N$ and satisfies $\lim_{N \rightarrow \infty} \nabla(N) =0$. Moreover, since ${{F_{{ {\cal A}_{\bs{\ell}}}}}\left( 2^{\cal R} \right)} \le 1$ and ${{W_{\bs{\ell}}}} \le 1$, the truncation error $\nabla(N)$ is in fact uniformly upper bounded by ${\sum\nolimits_{{\sum\nolimits_{k=1}^K  \ell_k \ge N+1}}  {{W_{\bs{\ell}}}} }$ which is irrespective of the rate $\cal R$ and converges to zero when $N$ increases. Thus strictly speaking, the truncation approximation in (\ref{eqn:F_G_x_trunc_ser}) actually admits a uniform convergence\cite[p147]{rudin1964principles}\footnote{We say that $f_n(x)$ converges uniformly to $f(x)$ on its domain $\cal D$ if, given $\varepsilon $, there exists an integer $N(\varepsilon )$ independent of $x$, such that $|f_n(x)-f(x)|\le \varepsilon $ for all $x \in \cal D$ whenever $n \ge N(\varepsilon )$\cite[p147]{rudin1964principles}.}. This uniform convergence guarantees the truncation approximation of outage probability with high accuracy for any $\cal R$. It is different from the convergence in MSE in \cite{shi2015analysis,shi2016inverse} which does not necessarily imply that the approximation error approaches to zero for arbitrary $\cal R$ \cite[p86]{adams2013continuous}. Therefore our truncation approximation is expected to perform better than those in \cite{shi2015analysis,shi2016inverse}, which will be further demonstrated in Section \ref{sec:num_res}.

As proved in Appendix \ref{app:uperb_err}, we further notice that the truncation error $\nabla(N)$ is bounded as 
\begin{align}\label{eqn:trun_err_fina_hyp}
\nabla(N)  
 &\le  {W_{\bf{0}}}{{F_{{ {\cal A}_{\bs{\ell}}},N}^{\rm max}}\left( 2^{\cal R} \right)} \xi \left( N \right) \triangleq \mathcal B_u, 
\end{align}
where ${{F_{{ {\cal A}_{\bs{\ell}}},N}^{\rm max}}\left( 2^{\cal R} \right)} = \mathop {\max }\nolimits_{\sum\nolimits_{k=1}^K  \ell_k = N + 1} \left( {{F_{{ {\cal A}_{\bs{\ell}}}}}\left( 2^{\cal R} \right)} \right)$, ${\xi \left( N \right)}$ is a decreasing function of $N$ as $\xi \left( N \right) ={\left( {\sum\nolimits_{k = 1}^K {{w_k}} } \right)^{N + 1}}\frac{{{{\left( m \right)}_{N + 1}}}}{{\left( {N + 1} \right)!}} {}_2{F_1}\left( {m + N + 1,1;N + 2;\sum\nolimits_{k = 1}^K {{w_k}} } \right)$, the notation $(\cdot)_n$ stands for Pochhammer symbol, and ${}_2F_1(a,b;c;z)=\sum\nolimits_{s = 0}^\infty  {\frac{{{{\left( a \right)}_s}{{\left( b \right)}_s}}}{{\Gamma \left( {c + s} \right)s!}}{z^s}} $ denotes hypergeometric function \cite[Eq.15.1.1]{olver2010nist}. For Rayleigh fading channels, i.e., $m=1$, the term $\xi \left( N \right)$ reduces to $\xi \left( N \right) = {W_{\bf{0}}}^{ - 1}{\left( {\sum\nolimits_{k = 1}^K {{w_k}} } \right)^{N + 1}}$ by using ${}_2{F_1}\left( {N + 2,1;N + 2;\sum\nolimits_{k = 1}^K {{w_k}} } \right) = {W_{\bf{0}}}^{ - 1}$ in \cite[Eq.1.40]{mathai2009h}. Together with ${{F_{{ {\cal A}_{\bs{\ell}}},N}^{\rm max}}\left( 2^{\cal R} \right)} \le 1$, the truncation error under Rayleigh fading channels is thus further bounded as $\nabla(N) \le \mathcal B_u \le {\left( {\sum\nolimits_{k = 1}^K {{w_k}} } \right)^{N + 1}}$. Since ${\sum\nolimits_{k = 1}^K {{w_k}} } < 1$, the truncation error under Rayleigh fading channels exponentially decays with $N$. It implies that the convergence speed of our truncation approach is fast and also justifies the effectiveness of our truncation approach.

In practice, the truncation order plays an important role in improving the approximation accuracy and an efficient selection of the truncation order is necessary. Specifically, for a given accuracy requirement $\varepsilon $, the minimal truncation order $N(\varepsilon)$ can be determined by setting the upper bound $\mathcal B_u$ in (\ref{eqn:trun_err_fina_hyp}) less than the maximum tolerable approximation error $\varepsilon $, such that $N(\varepsilon) = \mathop {\min } \left\{\left. N \ge 0 \right| \mathcal B_u \le \varepsilon   \right\}$.

The impact of truncation order on the approximation accuracy is then examined and the results are shown in Fig. \ref{fig:N_K}. Clearly, the truncated outage probability $\tilde p_{out,K}^N$ converges quite fast to the exact outage probability $p_{out,K}$ and the truncation order of $N=3$ is sufficient for a very good approximation of the outage probability. Plugging $N=3$ together with ${\cal R}=1$bps/Hz into (\ref{eqn:trun_err_fina_hyp}), the upper bounds $\mathcal B_u$ of four cases from I to IV can be calculated as  $7.0*10^{-5}$, $7.4*10^{-6}$, $1.4*10^{-4}$ and $3.7*10^{-8}$, respectively. They are very small and thus verify the accuracy of our truncation approach. By comparing the upper bound $\mathcal B_u$  of Case I with that of Case II, it can be observed that the increase of $K$ decreases $\mathcal B_u$ from $7.0*10^{-5}$ to $7.4*10^{-6}$, which means that the increase of $K$ would be beneficial to the reduction of the truncation error and thus improve the accuracy of our truncation approach. This can be roughly explained as follows. From the upper bound $\mathcal B_u$ in (\ref{eqn:trun_err_fina_hyp}), we can see that increasing $K$ has twofold impact on the upper bound $\mathcal B_u$. Specifically, the increase of $K$ reduces the first component ${W_{\bf{0}}}{{F_{{ {\cal A}_{\bs{\ell}}},N}^{\rm max}}\left( 2^{\cal R} \right)}$ of $\mathcal B_u$ on one hand, while increases the second component $\xi \left( N \right)$ on the other hand. Since the first component dominates particularly in a low-to-medium outage region which is the most concerned region for practical applications, the upper bound $\mathcal B_u$ can be viewed as a decreasing function of $K$.

\begin{figure}
  \centering
  \includegraphics[width=2.5in]{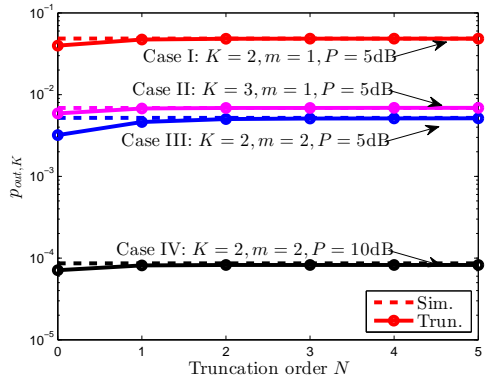}\\
  \caption{The impact of truncation order by setting $\lambda_1=\cdots=\lambda_K=0.8$, $\frac{{\sigma_1}^2}{\mathcal N_0}=\cdots=\frac{{\sigma_K}^2}{\mathcal N_0}=1$, $P_1=\cdots=P_K=P$ and $\mathcal R=1$bps/Hz.}\label{fig:N_K}
\end{figure}

\section{Asymptotic Outage Analysis}\label{sec:asy}
\subsection{Asymptotic Outage Probability}
Although the outage probability of HARQ-IR over time-correlated Nakagami fading channels can be exactly derived in closed-form as shown in (\ref{eqn:out_prob_def_hat}) and can be approximated as (\ref{eqn:F_G_x_trunc_ser}), they are still complex and hinder the extraction of meaningful insights. Fortunately, after analyzing the CDF ${F_{{{ {\cal A}}_{\bs{\ell}}}}}\left( x \right)$, we can find some special properties of ${F_{{{ {\cal A}}_{\bs{\ell}}}}}\left( x \right)$, which enable us to further conduct asymptotic outage analysis so that the expression of outage probability can be simplified with clear insights. To proceed with the analysis, we rewrite the transmission SNRs as
\begin{equation}\label{eqn:gamma_snr_rel}
\left( {\frac{{{P_1}}}{{{{\cal N}_0}}},\frac{{{P_2}}}{{{{\cal N}_0}}}, \cdots ,\frac{{{P_K}}}{{{{\cal N}_0}}}} \right) = \frac{{{P_{total}}}}{{{{\cal N}_0}}}\left( {\frac{{{P_1}}}{{{P_{total}}}},\frac{{{P_2}}}{{{P_{total}}}}, \cdots ,\frac{{{P_K}}}{{{P_{total}}}}} \right) = \gamma \bs \theta,
\end{equation}
where ${P_{total}} = {P_1} +  \cdots  + {P_K}$, $\gamma \triangleq \frac{{{P_{total}}}}{{{{\cal N}_0}}}$ denotes the total transmission SNR and $\bs \theta = ({\theta_1},\cdots,{\theta_K})\triangleq\left( {\frac{{{P_1}}}{{{P_{total}}}}, \cdots ,\frac{{{P_K}}}{{{P_{total}}}}} \right)$ represents the power allocation vector. With this definition of the transmission SNR, the CDF ${F_{{{ {\cal A}}_{\bs{\ell}}}}}\left( x \right)$ in (\ref{eqn:CDF_F_A_def_sec}) can be rewritten as shown in the following theorem.
\begin{theorem} \label{lemma:lemma_gener_fox_H}
The CDF ${F_{{{ {\mathcal A}}_{\bs{\ell}}}}}(x)$ can be rewritten as
\begin{multline}\label{eqn:der_A_1_CDF}
{F_{{{\cal A}_{\bs{\ell}}}}}(x) = {\gamma ^{-d_{{ {\cal A}_{\bs{\ell}}}}}}\prod\nolimits_{k = 1}^K {\frac{1}{{\Gamma \left( {m + {\ell_k}} \right)}}{{\left( {\frac{m}{{{\theta _k}{\sigma _k}^2\left( {1 - {\lambda _k}^2} \right)}}} \right)}^{m + {\ell_k}}}} \\
\times \sum\nolimits_{{n_1}, \cdots ,{n_K} = 0}^\infty  {\prod\nolimits_{k = 1}^K {\frac{1}{{{n_k}!}}{{\left( { - \frac{m}{{\gamma {\theta _k}{\sigma _k}^2\left( {1 - {\lambda _k}^2} \right)}}} \right)}^{{n_k}}}} } {g_{{\bf{n}} + {\bs{\ell}}}}\left( x \right),
\end{multline}
where $d_{{ {\cal A}_{\bs{\ell}}}} = mK + \sum\nolimits_{k = 1}^K { { {\ell_k}} }$, ${\bf n} = (n_1,\cdots,n_K)$, ${{\bs{\ell}}} = (\ell_1,\cdots,\ell_K)$ and ${g_{{\bs{\ell}}}}\left( x \right)$ is given by
\begin{align}\label{eqn:sec_def_g_fun}
{g_{\bs{\ell}}}\left( x \right) &= \int\nolimits_{\prod\nolimits_{k = 1}^K {\left( {1 + {t_k}} \right)}  \le x} {\prod\limits_{k = 1}^K {{t_k}^{m + {\ell_k} - 1}} d{t_1} \cdots d{t_{K - 1}}d{t_K}}\\
\label{eqn:g_0_0_der_meij_rem}
&= \prod\limits_{k = 1}^K {\Gamma \left( {m + {\ell_k}} \right)} G_{K + 1,K + 1}^{0,K + 1}\left( {\left. {\begin{array}{*{20}{c}}
{1,1 + {\ell_1} + m, \cdots ,1 + {\ell_K} + m}\\
{1, \cdots ,1,0}
\end{array}} \right|x} \right),
\end{align}
with $G_{p,q}^{m,n}\left( {\cdot|x} \right)$ denoting Meijer-G function \cite[Eq.9.301]{gradshteyn1965table}. 
\end{theorem}
\begin{proof}
Please see Appendix \ref{app:lemma_gener_fox_H}.
\end{proof}


From Theorem \ref{lemma:lemma_gener_fox_H}, when the transmission SNR is high, i.e., $\gamma \to \infty$, the CDF ${F_{{{ {\mathcal A}}_{\bs{\ell}}}}}(x)$ can be further expressed as
\begin{equation}\label{eqn:der_A_1_asyCDF}
{F_{{{\cal A}_{\bs{\ell}}}}}(x) = {\gamma ^{ - {d_{{\mathcal A_{\bs{\ell}}}}}}}\prod\limits_{k = 1}^K {\frac{1}{{\Gamma \left( {m + {\ell_k}} \right)}}{{\left( {\frac{m}{{{\theta _k}{\sigma _k}^2\left( {1 - {\lambda _k}^2} \right)}}} \right)}^{m + {\ell_k}}}} {g_{\bs{\ell}}}\left( x \right) + o\left( {{\gamma ^{ - {d_{{\mathcal A_{\bs{\ell}}}}}}}} \right),
\end{equation}
where $o(\cdot)$ refers to the little-O notation, and $f(\gamma) \in o(\phi(\gamma))$ provided that $\lim \limits_{\gamma \to \infty}f(\gamma)/\phi(\gamma) = 0$. Based on (\ref{eqn:der_A_1_asyCDF}), the following property of the CDF ${F_{{{ {\mathcal A}}_{\bs{\ell}}}}}(x)$ holds since $d_{{ {\cal A}_{\bs{\ell}}}} = mK + \sum\nolimits_{k = 1}^K { { {\ell_k}} } > d_{{ {\cal A}_{\bf{0}}}} = mK$ for ${\bs{\ell} \ne \bf{0}}$.
\begin{lemma}\label{the:cdf_pdf_relatoin_hat}
As $\gamma \to \infty$, the ratio of ${{F_{{ {\mathcal A}_{\bs{\ell}}}}}\left( x \right)}$ to ${{F_{{ {\mathcal A}_{\bf{0}}}}}\left( x \right)}$ satisfies
\begin{equation}\label{eqn:F_AI_F_A0_rel_hat}
\frac{{{F_{{ {\cal A}_{\bs{\ell}}}}}\left( x \right)}}{{{F_{{ {\cal A}_{\bf{0}}}}}\left( x \right)}} = o\left( 1 \right),\, {{{\bs{\ell}}} \in {{\mathbb N}_0}^K} \,{\rm and }\, {\bs{\ell} \ne \bf{0}}.
\end{equation}
\end{lemma}
\begin{proof}
From (\ref{eqn:der_A_1_asyCDF}), we have
  \begin{align}\label{eqn:lemma2_proof}
\mathop {\lim }\limits_{\gamma  \to \infty } \frac{{{F_{{\mathcal A_{\bs \ell} }}}\left( x \right)}}{{{F_{{\mathcal A_{\bf{0}}}}}\left( x \right)}} 
 &= \prod\limits_{k = 1}^K {\frac{{\Gamma \left( m \right)}}{{\Gamma \left( {m + {\ell _k}} \right)}}{{\left( {\frac{m}{{{\theta _k}{\sigma _k}^2\left( {1 - {\lambda _k}^2} \right)}}} \right)}^{{\ell _k}}}} \frac{{{g_{\bs \ell} }\left( x \right)}}{{{g_{\bf{0}}}\left( x \right)}}\mathop {\lim }\limits_{\gamma  \to \infty } {\gamma ^{d_{{{\cal A}_{\bf 0} }} - {d_{{{\cal A}_{\bs \ell} }}}}} = 0,
\end{align}
where the last equality holds because of ${d_{{{\cal A}_{\bf 0} }} - {d_{{{\cal A}_{\bs \ell} }}}} < 0$ when $\bs \ell \ne \bf 0$. Then by using the little-O notation, (\ref{eqn:F_AI_F_A0_rel_hat}) directly follows.
\end{proof}
By applying Lemma \ref{the:cdf_pdf_relatoin_hat} into (\ref{eqn:out_prob_def_hat}), as $\gamma \to \infty$, the outage probability ${p_{out,K}}$ can be rewritten as \begin{align}\label{eqn:cdf_F_shif_gam_pro_asy}
{p_{out,K}} &= {W_{\bf{0}}}{F_{{{ {\mathcal A}}_{\bf{0}}}}}\left( 2^{\mathcal R} \right) + \sum\nolimits_{{\ell_1} +  \cdots  + {\ell_K} > 0}  {{W_{\bs{\ell}}}{F_{{{ {\mathcal A}}_{\bs{\ell}}}}}\left( 2^{\mathcal R} \right)} \notag \\
&= {W_{\bf{0}}}{F_{{{ {\mathcal A}}_{\bf{0}}}}}\left( 2^{\mathcal R} \right)\left( {1 + \frac{1}{{{W_{\bf{0}}}}}\sum\nolimits_{{\ell_1} + \cdots  + {\ell_K} > 0}  {{W_{\bs{\ell}}}\frac{{{F_{{{ {\mathcal A}}_{\bs{\ell}}}}}\left( 2^{\mathcal R} \right)}}{{{F_{{{ {\mathcal A}}_{\bf{0}}}}}\left( 2^{\mathcal R} \right)}}} } \right)\notag\\
& = {W_{\bf{0}}}{F_{{{ {\mathcal A}}_{\bf{0}}}}}\left( 2^{\mathcal R} \right)\left( {1 + \frac{1}{{{W_{\bf{0}}}}}\sum\nolimits_{{\ell_1} +  \cdots  + {\ell_K} > 0}  {{W_{\bs{\ell}}}o\left( 1  \right)} } \right) = {W_{\bf{0}}}{F_{{{ {\mathcal A}}_{\bf{0}}}}}\left( 2^{\mathcal R} \right)\left( {1 + o\left( 1 \right)} \right),
\end{align}
where the last equality holds since $W_{\bs{\ell}}$ is irrelevant to the transmission SNR $\gamma$ and $\sum\nolimits_{{\ell_1} + \cdots  + {\ell_K} > 0}  {{W_{\bs{\ell}}}} = 1-W_{\bf 0} < 1$. Now putting (\ref{eqn:der_A_1_asyCDF}) into (\ref{eqn:cdf_F_shif_gam_pro_asy}) and neglecting the high order terms $o\left( 1 \right)$ and $o\left( {{\gamma ^{ - {d_{{\mathcal A_{\bs{\ell}}}}}}}} \right)$ when $\gamma \to \infty$, the outage probability can be asymptotically approximated as
%
\begin{align}\label{eqn:der_hat_G_CDF_asy_fur}
{p_{out,K}} 
 &\simeq  {W_{\bf{0}}}{\gamma ^{ - mK}}{g_{\bf{0}}}\left(  2^{\cal R} \right)\prod\limits_{k = 1}^K {\frac{1}{{\Gamma \left( m \right)}}{{\left( {\frac{m}{{{\theta _k}{\sigma _k}^2\left( {1 - {\lambda _k}^2} \right)}}} \right)}^m}} \triangleq {p_{out\_{asy},K}}. 
\end{align}

Substituting (\ref{eqn:W_l_def_sec}) into (\ref{eqn:der_hat_G_CDF_asy_fur}), under high SNR regime, i.e., as $\gamma \to \infty$, the asymptotic outage probability ${p_{out\_{asy},K}}$ can be factorized as
\begin{align}\label{eqn:der_hat_G_CDF_asy_furre}
{p_{out\_{asy},K}} &= \underbrace {\prod\limits_{k = 1}^K {\frac{1}{{\Gamma \left( m \right)}}{{\left( {\frac{m}{{{\theta _k}{\sigma _k}^2}}} \right)}^m}} }_{\triangleq \zeta(\bs \theta)} \underbrace {{{\left( {\ell \left( {\bs \lambda ,K} \right)} \right)}^{ - m}}}_{\triangleq \varrho(\bs \lambda)} {\left( {\underbrace {{{\left( {{g_{\bf{0}}}\left(  2^{\cal R} \right)} \right)}^{ - \frac{1}{mK}}}}_{\triangleq {C}(\mathcal R)}\gamma } \right)^{ - mK}},
\end{align}
where $\ell \left( {{\bs{\lambda }},K} \right)$ is explicitly given as
 \begin{equation}\label{eqn:ell_def}
\ell \left( {{\bs{\lambda }},K} \right) = \left( {1 + \sum\nolimits_{k = 1}^K {\frac{{{\lambda _k}^2}}{{1 - {\lambda _k}^2}}} } \right)\prod\nolimits_{k = 1}^K {\left( {1 - {\lambda _k}^2} \right)},
\end{equation}
$\zeta (\bs \theta)$ and $\varrho(\bs \lambda)$ quantify the impacts of transmission power allocation and channel time correlation and thus are regarded as power allocation impact factor and time correlation impact factor, respectively, and $C(\mathcal R)$ is termed as coding and modulation gain since it depends on the information transmission rate $\mathcal R$ which eventually is determined by the coding rate, modulation scheme and symbol transmission rate.

\subsection{Discussions}
Clearly from (\ref{eqn:der_hat_G_CDF_asy_furre}), the outage performance is determined by the number of transmissions, transmission power allocation, channel time correlation and information transmission rate. Their impacts will be thoroughly investigated through the analysis of the terms $\zeta(\bs \theta) $, $\varrho(\bs \lambda)$ and $C(\mathcal R)$. In addition, outage probability of HARQ-IR under quasi-static fading channels is particularly discussed as a special case of time-correlated fading channels.
\subsubsection{Diversity Order}
Diversity order indicates the number of degrees of freedom in communication systems. Roughly speaking, it is equivalent to the number of independently faded paths that a transmitted signal experiences. Specifically, it is defined as the slope of the outage probability against the transmission SNR on a log-log scale as \cite{chelli2014performance}
\begin{equation}\label{eqn:diver_order_def}
d = -\mathop {\lim }\limits_{\gamma  \to \infty } \frac{{\log \left( {{ p_{out,K}}} \right)}}{{\log \left( \gamma  \right)}}.
\end{equation}
Putting (\ref{eqn:der_hat_G_CDF_asy_furre}) into (\ref{eqn:diver_order_def}), the diversity order of HARQ-IR over time-correlated Nakagami-$m$ fading channels directly follows as $d=mK$. Noticing that a Nakagami-$m$ fading channel can be viewed as a set of $m$ parallel independent Rayleigh fading channels, the maximum number of independent fading channels in a HARQ-IR system with a maximum number of $K$ transmissions is $mK$. In other words, channel time correlation would not degrade the diversity order and full diversity can be achieved by HARQ-IR even under time-correlated fading channels. This result is consistent with those in \cite{shi2015analysis,shi2016inverse}, and demonstrates the correctness of our asymptotic outage analysis.

\subsubsection{Power Allocation Impact Factor $\zeta (\bs \theta)$}
It is clear from (\ref{eqn:der_hat_G_CDF_asy_furre}) that power allocation impact factor $\zeta(\bs \theta)$ characterizes the impact of power allocation on outage probability. Specifically, $\zeta(\bs \theta)$ is an inverse power function of the the product of power allocation factors ${\prod\nolimits_{k = 1}^K {{\theta_k}} }$. It would decrease as the product of power allocation factors $\bs \theta$ increases, which eventually results in the improvement of outage performance. With the definition of $\theta_k=\frac{{{P_k}}}{{{P_{total}}}}$, we also can conclude that the asymptotic outage probability is an inverse power function of the product of the transmission powers in all HARQ rounds. Notice that this clear quantitative relationship between the power allocation factors and the outage probability hasn't been discovered even under independent fading channels. With this quantitative relationship, optimal power allocation to achieve various objectives is enabled. Taking the energy-limited outage minimization as an example, the power allocation problem can be formulated as
\begin{equation}\label{eqn:opt_prob_simp}
\begin{array}{*{20}{cl}}
{\mathop {\rm min}\limits_{{P_1},\cdots,{P_K}} }&{p_{out,K}}\\
{{\rm{s.}}\,{\rm{t.}}}&{{\sum\nolimits_{k = 1}^K {{P_k}{p_{out,k - 1}}}} \le P_{\rm given} }\\
{}&{{P_k} \ge 0,\quad{\rm for}\quad 0 \le k \le K},\\
\end{array}
\end{equation}
where ${\sum\nolimits_{k = 1}^K {{P_k}{p_{out,k - 1}}}}$ refers to the average transmission energy normalized to the codeword length \cite{makki2013green},  $p_{out,0}=1$ and $P_{\rm given}$ is the average energy constraint. By substituting $P_k= {\theta_k} P_{total}$ and using the asymptotic outage probability ${p_{out\_{asy},k}}$ as the approximation of $p_{out,k}$, the optimization problem is reduced to the maximization of the product of power allocation factors ${\prod\nolimits_{k = 1}^K {{\theta_k}} }$ subject to certain constraints. Accordingly, (\ref{eqn:opt_prob_simp}) can be easily converted into the generalized power optimization problem in \cite[Eq.6]{chaitanya2016energy}. Karush-Khun-Tucker conditions can then be applied to derive the optimal solution in closed-form as shown in \cite[Eq.15]{chaitanya2016energy}.

\subsubsection{Time Correlation Impact Factor $\varrho(\bs \lambda)$}
The time correlation impact factor $\varrho(\bs \lambda)$ quantifies the impact of channel time correlation on the outage probability. Notice that when only one transmission is allowed, i.e., $K=1$, no time correlation is involved and we always have $\ell \left( {{\bs{\lambda }},K} \right)=1$ and $\varrho(\bs \lambda)=1$. When $K>1$, the time correlation impact factor $\varrho(\bs \lambda)$ has one property as shown in the following lemma. For notational convenience, we define a partial ordering for two vectors ${\bf x}= (x_1, x_2, \cdots, x_K), {\bf y}= (y_1, y_2, \cdots, y_K) \in \mathbb{R}^K$ as ${\bf x} \preceq \bf y$ if $x_i \le y_i, i = 1,2,\cdots,K$.
%
%
\begin{lemma}\label{cor:time_corr}
Given $K>1$ and two time correlation vectors $\bs{\lambda}_1 \preceq \bs{\lambda}_2$, we have
\begin{equation}\label{eqn:lam_ell_con}
1 \ge \ell \left( {{\bs{\lambda }_1},K} \right) \ge \ell \left( {{\bs{\lambda }_2},K} \right),
\end{equation}
\begin{equation} \label{eqn:varrho}
1 \le \varrho(\bs \lambda_1) \le \varrho(\bs \lambda_2).
\end{equation}
The left equalities in (\ref{eqn:lam_ell_con})-(\ref{eqn:varrho}) hold if and only if $\bs \lambda_1 = \bf 0$, while the right equalities in (\ref{eqn:lam_ell_con})-(\ref{eqn:varrho}) hold if and only if $\bs \lambda_1 = \bs \lambda_2$ .
\end{lemma}
\begin{proof}
Please see Appendix \ref{app:proo_cor}.
\end{proof}

It can be concluded from (\ref{eqn:der_hat_G_CDF_asy_furre}) and Lemma \ref{cor:time_corr} that although the time correlation does not affect the diversity order, the increase of time correlation coefficients would cause the increases of $\varrho(\bs \lambda)$ and the outage probability, thus resulting in the degradation of outage performance.
%

\subsubsection{Coding and Modulation Gain $C(\mathcal R)$}
As defined in (\ref{eqn:der_hat_G_CDF_asy_furre}), $C(\mathcal R)=\left({g_{{\bf{0}}}}\left( 2^{\cal R} \right)\right)^{-\frac{1}{d}}$. For a given outage probability, the increase of $C(\mathcal R)$ can result in the reduction of the SNR $\gamma $. In other words, $C(\mathcal R)$ can quantify the amount of SNR reduction for a given outage probability under certain coding and modulation scheme. It is thus termed as coding and modulation gain \cite{goldsmith2005wireless}. After analyzing the function ${g_{{\bf{0}}}}\left( 2^{\mathcal R} \right)$, we have the following property of the coding and modulation gain.
\begin{lemma}\label{the:coding_modulation_gain}
The function ${g_{{\bf{0}}}}\left( 2^{\mathcal R} \right)$ is a monotonically increasing function of the information transmission rate $\mathcal R$, and is convex with respect to $\mathcal R$ for any fading order $m \ge 1$, and thus the coding and modulation gain $C(\mathcal R)$ is a monotonically decreasing function of the information transmission rate $\mathcal R$.
\end{lemma}
\begin{proof}
Please see Appendix \ref{app:coding_gain}.
\end{proof}

It follows from (\ref{eqn:der_hat_G_CDF_asy_furre}) and Lemma \ref{the:coding_modulation_gain} that the asymptotic outage probability is a monotonically increasing and convex function of the information transmission rate $\mathcal R$ when $m \ge 1$. In order to achieve a desired performance, the information transmission rate should be properly chosen. Owing to the simple analytical expression in (\ref{eqn:der_hat_G_CDF_asy_furre}), the optimal rate can be easily found. Taking the maximization of the long term average throughput (LTAT) given an allowable outage constraint $\epsilon $ as an example, the rate selection problem can be formulated as
\begin{align}\label{eqn_op}
{\mathop {{\rm{max}}}\limits_{\mathcal R \in \left\{\mathcal R \in \mathbb R_+: {p_{out,K}} \le \epsilon \right\}} }&{\quad \bar {\mathcal T}  = \frac{{\mathcal R\left( {1 - {p_{out,K}}} \right)}}{{\sum\nolimits_{k = 0}^{K - 1} {{p_{out,k}}} }}},
\end{align}
where $\bar {\mathcal T}$ denotes the LTAT. Clearly, the numerator ${\mathcal R\left( {1 - {p_{out,K}}} \right)}$ and the denominator ${ \sum\nolimits_{k = 0}^{K - 1} {{p_{out,k}}} }$ in (\ref{eqn_op}) are concave and convex with respect to $\mathcal R$, respectively, while the feasible region $\left\{\mathcal R \in \mathbb R_+: {p_{out,K}} \le \epsilon \right\}$ is a convex set, when $m \ge 1$. As shown in \cite{dinkelbach1967nonlinear}, the optimization problem in (\ref{eqn_op}) is a concave fractional programming problem and the globally optimal solution can be easily found using the techniques proposed in \cite{dinkelbach1967nonlinear}.


\subsubsection{Quasi-Static Fading channels}

As aforementioned, the probability distribution in Section \ref{sec:exa} and the asymptotic outage probability in (\ref{eqn:der_hat_G_CDF_asy_furre}) are not applicable to the case of quasi-static fading channels with ${\boldsymbol{\lambda }} = {\bf 1}$. Particularly, under quasi-static fading channels, the channel coefficients are constant among multiple transmissions, i.e., $|h_1|=\cdots=|h_K| \triangleq |h| \sim {\rm Nakagami}(m,{\sigma}^2)$. Assuming constant transmission powers $P_1 = \cdots= P_K \triangleq P$, that is, $\theta_1 = \cdots= \theta_K \triangleq \theta$, the outage probability ${p_{out,K}}$ under quasi-static fading channels is readily obtained as
\begin{equation}\label{eqn:out_prob_quasi_static}
{p_{out,K}} = \Pr \left( {{{\log }_2}{{\left( {1 + \frac{{P}{{\left| {{h}} \right|}^2}}{\mathcal N_0}} \right)}^K} < \mathcal R} \right) = \frac{1}{{\Gamma \left( m \right)}}\Upsilon \left( {m,\frac{{m\mathcal N_0\left( {{2^{\mathcal R/K}} - 1} \right)}}{{{P}{\sigma}^2}}} \right).
\end{equation}
By applying \cite[Eq.8.354.1]{gradshteyn1965table} into (\ref{eqn:out_prob_quasi_static}), the outage probability can be rewritten as
\begin{align}\label{eqn:out_prob_quasi_asy}
{p_{out,K}}&= \frac{1}{{\Gamma \left( m \right)}}\sum\limits_{n = 0}^\infty  {\frac{{{{\left( { - 1} \right)}^n}}}{{n!\left( {m + n} \right)}}{{\left( {\frac{{m\left( {{2^{{\cal R}/K}} - 1} \right)}}{{\gamma {\theta{\sigma}^2}}}} \right)}^{m + n}}} \notag \\
&= {\left( {\frac{m}{{{\theta}{\sigma}^2}}} \right)^m}{\left( {\frac{{{{\left( {\Gamma \left( {m + 1} \right)} \right)}^{1/m}}}}{{{2^{{\cal R}/K}} - 1}}\gamma } \right)^{ - m}} + o\left( {{\gamma ^{ - m}}} \right).
\end{align}
Then the asymptotic outage probability under this quasi-static fading channels can be written as ${p_{out\_{asy},K}} = \zeta(\bs \theta) \left( C(\mathcal R \right ) \gamma)^{-m}$, where the power allocation impact factor is given as $\zeta(\bs \theta) = {\left( {\frac{m}{{{\theta{\sigma}^2}}}} \right)^{ m}}$, while the coding and modulation gain becomes $C(\mathcal R) = \frac{{{{\left( {\Gamma \left( {m + 1} \right)} \right)}^{1/m}}}}{{{2^{{\mathcal R}/K}} - 1}}$ and the diversity order reduces to $d=m$ due to the full correlation of fading channels. It means that no time diversity can be achieved from multiple transmissions under quasi-static fading channels.

\section{Numerical Results And Optimal Design} \label{sec:num_res}
The analytical results are now verified and optimal system design is discussed in this section. For illustration, we take systems with equicorrelated channels (i.e., $\bs \lambda_{\rm eq}=(\rho,\cdots,\rho)$) \cite{alouini2001sum,aalo1995performance,chen2004distribution} and unit Nakagami spread ${\sigma_1}^2=\cdots={\sigma_K}^2=1$ as examples, unless otherwise indicated. 
In the following numerical analysis, the exact outage probability is approximated by (\ref{eqn:F_G_x_trunc_ser}) with the truncation order set as $N=3$ and the involved generalized Fox's H function is efficiently calculated by the MATHEMATICA{\textregistered} program in \cite{yilmaz2010outage} with a properly chosen Mellin-Barnes contour.


\subsection{Performance Evaluation}
In Fig. \ref{fig:out_s}, the outage probability $p_{out,K}$ is plotted against the transmission SNR $\gamma$ by setting $K=4$ and $\mathcal R=4$bps/Hz. The approximated results based on  polynomial fitting technique \cite{shi2015analysis} and inverse moment matching method \cite{shi2016inverse} are also presented for comparison. Notice that polynomial fitting technique in \cite{shi2015analysis} is proposed for HARQ-IR under Rayleigh fading channels and thus its result only for Rayleigh fading (i.e., $m=1$) is shown in Fig. \ref{fig:out_s}. It is clear that the approximated results (\ref{eqn:F_G_x_trunc_ser}) coincide with the simulation results, while the asymptotic results (\ref{eqn:der_hat_G_CDF_asy_furre}) approach to the approximated/simulation results under high SNR regime. However, under the considered scenarios, neither \cite{shi2015analysis} nor \cite{shi2016inverse} can provide an accurate approximation under high SNR regime due to the fact that those proposed methods can only guarantee their convergence in MSE. In addition, the outage probabilities decrease with the increase of the transmission SNR $\gamma$ and the decreasing rate becomes larger when the fading order $m$ increases, because the diversity order is $mK$ and $p_{out\_asy,K}$ is directly proportional to $\gamma^{-mK}$. As expected, channel time correlation has a detrimental impact on outage probability. For a given fading order $m$, the outage probability curves corresponding to different correlations become parallel as $\gamma$ increases due to the same diversity order which is irrelevant to the correlation coefficient. These numerical results thus validate the results in Section \ref{sec:asy}.


\begin{figure}
  \centering
  \includegraphics[width=2.5in]{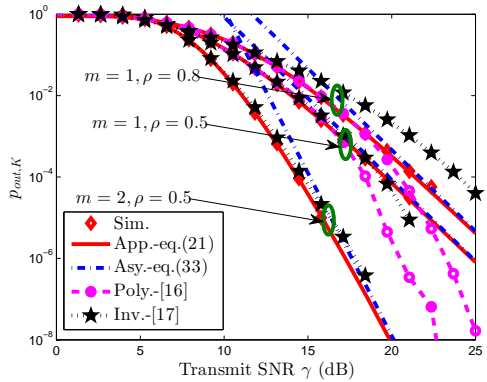}\\
  \caption{Outage probability $p_{out,K}$ versus transmission SNR $\gamma$.}\label{fig:out_s}
\end{figure}

To further investigate the effect of time correlation, Fig. \ref{fig:corr} depicts time correlation impact factors $\varrho(\bs \lambda)$ under two different channel correlation models, i.e., equal correlation and exponential correlation. For exponential correlation model, the correlation coefficients follow as $\bs \lambda_{\rm exp}=(\rho,\rho^2,\cdots,\rho^K)$ \cite{aalo1995performance,kim2011optimal}. It can be seen that the time correlation impact factor $\varrho(\bs \lambda)$ increases as $\rho$ increases, thus causing the degradation of the outage performance. The curves become steeper with the increase of $\rho$, which means the outage performance degradation would become more significant when time correlation is high. By comparing the time correlation impact factors $\varrho(\bs \lambda_{\rm eq})$ and $\varrho(\bs \lambda_{\rm exp})$, we can find that $\bs\lambda_{\rm exp} \preceq \bs\lambda_{\rm eq} $ and  $\varrho(\bs \lambda_{\rm exp}) \le \varrho(\bs \lambda_{\rm eq})$, which is consistent with Lemma \ref{cor:time_corr}. In addition, it is found that the increase of the maximum number of transmissions $K$ will lead to the increase of the time correlation impact factor $\varrho(\bs \lambda)$ no matter under equal correlation model or under exponential correlation model. This is because that $\varrho(\bs \lambda_{K}) = \varrho((\bs \lambda_K,0))$ and $(\bs \lambda_K,0) \preceq (\bs \lambda_K,\lambda_{K+1})$. It follows from Lemma \ref{cor:time_corr} that $\varrho(\bs \lambda_K) \le \varrho((\bs \lambda_K,\lambda_{K+1}))$ and the equality holds if and only if $\lambda_{K+1}=0$. However, when $K$ increases, the increase of the time correlation impact factor will be offset by the decrease of the term $\left( C(\mathcal R \right ) \gamma)^{-mK}$ and thus would not cause the degradation of the outage performance.

\begin{figure}
  \centering
  \includegraphics[width=2.5in]{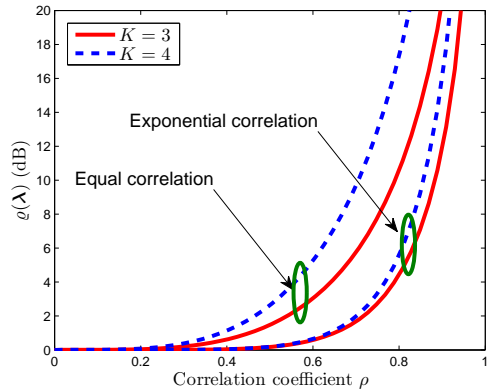}
  \caption{Time correlation impact factor with $m=2$.}\label{fig:corr}
\end{figure}

Finally, Fig. \ref{fig:coding_gain} illustrates the impacts of the information transmission rate $\cal R$ and the fading order $m$ on the coding and modulation gain $C(\cal R)$ under both time-correlated fading channels and quasi-static fading channels. Clearly, the coding and modulation gain $C(\cal R)$ decreases with the increase of the rate $\mathcal R$, which has already been proved in Lemma \ref{the:coding_modulation_gain}. Additionally, the coding and modulation gain over quasi-static fading channels is superior to the gain over time-correlated fading channels. The reason behind this is that multiple HARQ transmissions under quasi-static fading channels can be viewed as one single transmission with lower coding rate due to the same channel realization experienced, thus leading to the improvement of coding and modulation gain. However, such coding and modulation gain improvement is achieved in sacrifice of time diversity, that is, the diversity order of quasi-static fading channels reduces to $m$. Moreover, Fig. \ref{fig:coding_gain} shows that the coding and modulation gain $C(\mathcal R)$ could benefit from the increase of fading order $m$. 
\begin{figure}
  \centering
  \includegraphics[width=2.5in]{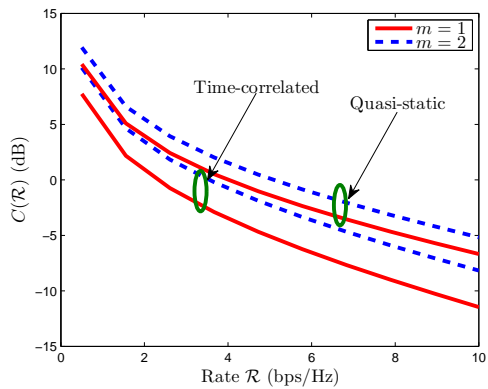}\\
  \caption{Coding and modulation gain $C(\mathcal R)$ versus information transmission rate $\mathcal R$ for systems with $K=4$.}\label{fig:coding_gain}
\end{figure}
\subsection{Optimal System Design}
With the simple expression, the asymptotic result in (\ref{eqn:der_hat_G_CDF_asy_furre}) would facilitate optimal system design for HARQ-IR with low complexity\footnote{Specifically, the asymptotic outage probability (\ref{eqn:der_hat_G_CDF_asy_furre}) can be easily computed by a few arithmetical operations. Its computational complexity is much lower than the simulation results which require the calculations on a large number of channel realizations. For example, in order to obtain the outage probability of $10^{-5}$, more than $10^6$ channel realizations are required to be simulated, which definitely causes very high computational complexity.}. Here optimal power allocation in (\ref{eqn:opt_prob_simp}) and optimal rate selection in (\ref{eqn_op}) are particularly investigated as examples.

The optimal power allocation (OPA) schemes in (\ref{eqn:opt_prob_simp}) designed based on approximated outage probability (\ref{eqn:F_G_x_trunc_ser}) and asymptotic outage probability (\ref{eqn:der_hat_G_CDF_asy_furre}) are first compared with an optimal equal power allocation (OEPA) scheme where the outage probability is minimized as (\ref{eqn:opt_prob_simp}) with additional equal power allocation constraint (i.e., $P_1=\cdots=P_K$). The achieved optimal outage probabilities $p_{out,K}^*$ for a system with $K=2$, $\mathcal R=2$bps/Hz and $\rho=0.5$ are shown in Fig. \ref{fig:ppa_epa_comp}. It is clear that the OPA solutions found based on approximated outage probability and asymptotic outage probability agree well and can lead to similar outage probability, which justifies the adoption of the asymptotic outage probability for system design. Notice that the OPA solution based on asymptotic outage probability can be easily found in closed-form as shown in \cite{chaitanya2016energy} and its computational complexity is significantly reduced compared to the optimization based on approximated outage probability which requires exhaustive search with a large number of numerical computations. Moreover, OPA scheme performs better than OEPA since the transmission powers in all HARQ rounds are optimized.

\begin{figure}
  \centering
  \includegraphics[width=2.5in]{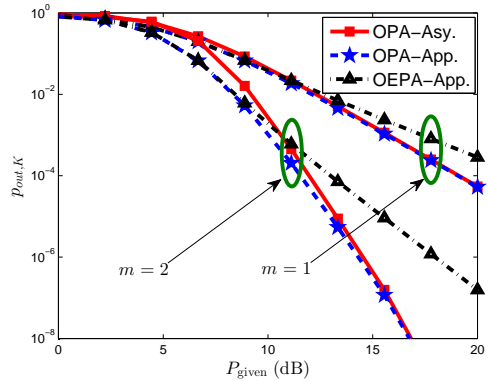}\\
  \caption{Comparison between OPA and OEPA.}\label{fig:ppa_epa_comp}
\end{figure}

Fig. \ref{fig:opt_rate} then illustrates the maximal LTAT $\bar {\mathcal T}$ achieved through optimal rate selection in (\ref{eqn_op}) by setting $m=1$ and $K=4$. Similarly, the results based on the approximated outage probability and the asymptotic outage probability match well, further validating our asymptotic results. It is also observed that the maximal LTAT increases when the outage constraint $\epsilon $ increases and/or the channel time correlation reduces.
\begin{figure}
  \centering
  \includegraphics[width=2.5in]{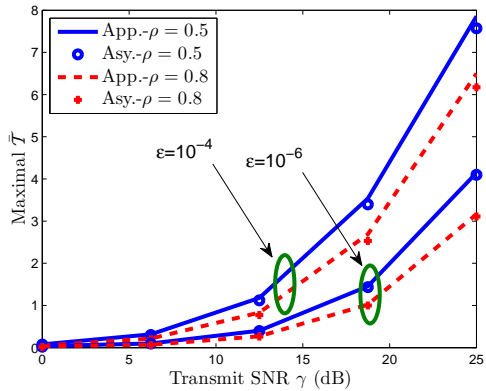}\\
  \caption{Maximal LTAT $\bar {\mathcal T}$ versus transmission SNR $\gamma$.}\label{fig:opt_rate}
\end{figure}
\section{Conclusions} \label{sec:con}
In this paper, asymptotic outage analysis has been conducted to thoroughly investigate the impacts of channel time correlation, transmit powers and information transmission rate on the performance of HARQ-IR over time correlated Nakagami fading channels. Clear insights have been discovered. Particularly, it has been revealed that the power allocation impact factor is an inverse power function of the product of power allocation factors, the time correlation impact factor is an increasing function of the channel time correlation coefficients, while the coding and modulation gain is a monotonically decreasing function of the information transmission rate. Therefore, high product of power allocation factors, low channel time correlation and low information transmission rate are favorable for improving outage performance. The simple form and special properties of asymptotic outage probability would effectively facilitate the optimal system design to achieve various objectives, e.g., optimal power allocation to minimize outage probability and optimal rate selection to maximize the LTAT.

\appendices
\section{The derivation of ${f_{\left. G_K \right|T}}(x|t)$}\label{app:phi_derivation}
The Mellin transform with respect to the conditional PDF of $G_K$ given $T=t$, ${f_{\left. G_K \right|T}}(x|t)$, can be written as
\begin{equation}\label{eqn_mellin_Rk1_def}
\left\{ {\mathcal M{f_{\left. G_K \right|T}}} \right\}\left( s \right) = {\rm E}\left\{ {\left. {{G_K^{s - 1}}} \right|T=t} \right\} \triangleq \phi \left( {\left. s \right|t} \right).
\end{equation}
Due to the independence of SNRs $\bs \gamma $ given $T=t$ and with the definition of $G_K \triangleq \prod\nolimits_{k = 1}^K {\left( {1 + {\gamma _k}} \right)}$, $\phi \left( {\left. s \right|t} \right)$ can be rewritten as
\begin{equation}\label{eqn_mellin_Rk}
\phi \left( {\left. s \right|t} \right) = \prod\nolimits_{k = 1}^K {{\rm E}\left\{ {\left. {{{{{\left(1 + \gamma_k\right)}} }^{s - 1}}} \right|t} \right\}}.
\end{equation}
Herein, with the conditional PDF ${{f_{\left. {{\gamma _k}} \right|T}}\left( {\left. {{x_k}} \right|t} \right)}$ in (\ref{eqn:cond_pdf_r_k_rew}), ${\rm E}\left\{ {\left. {{{\left( {1 + {\gamma_k}} \right)}^{s - 1}}} \right|t} \right\}$ can be derived as
\begin{equation}\label{eqn_mellin_r_1_p}
{\rm{E}}\left\{ {\left. {{{\left( {1 + {\gamma_k}} \right)}^{s - 1}}} \right|t} \right\} 
= \frac{{{e^{ - \frac{{{u_k}{\lambda _k}^2t}}{{{\Omega _k}}}}}}}{{\Gamma \left( m \right){{\left( {{\Omega _k}} \right)}^m}}}\int\nolimits_0^\infty  {{{\left( {1 + {x_k}} \right)}^{s - 1}}{x_k}^{m - 1}{e^{ - \frac{{{x_k}}}{{{\Omega _k}}}}}{}_0{F_1}\left( {;m;{{\left( {\frac{{\sqrt {{u_k}{\lambda _k}^2t} }}{{{\Omega _k}}}} \right)}^2}{x_k}} \right)d{x_k}}.
\end{equation}
By using the series expansion of the confluent hypergeometric limit function \cite[Eq.16.2.1]{olver2010nist}, (\ref{eqn_mellin_r_1_p}) is further derived as
\begin{align}\label{eqn_expre_mellin_r_expansion}
{\rm{E}}\left\{ {\left. {{{\left( {1 + \gamma_k} \right)}^{s - 1}}} \right|t} \right\} 
 &= \frac{{{e^{ - \frac{{{u_k}{\lambda _k}^2t}}{{{\Omega _k}}}}}}}{{{{\left( {{\Omega _k}} \right)}^m}}}\sum\limits_{{\ell_k} = 0}^\infty  {\frac{1}{{{\ell_k}!\Gamma \left( {m + {\ell_k}} \right)}}{{{{\left( {\frac{{\sqrt {{u_k}{\lambda _k}^2t} }}{{{\Omega _k}}}} \right)}^{2{\ell_k}}}}}} \int\nolimits_0^\infty  {{{\left( {1 + {x_k}} \right)}^{s - 1}}{x_k}^{m + {\ell_k} - 1}{e^{ - \frac{{{x_k}}}{{{\Omega _k}}}}}d{x_k}} \notag \\
 &= \frac{{{e^{ - \frac{{{u_k}{\lambda _k}^2t}}{{{\Omega _k}}}}}}}{{{{\left( {{\Omega _k}} \right)}^m}}}\sum\limits_{{\ell_k} = 0}^\infty  {\frac{1}{{{\ell_k}!}}{{{{\left( {\frac{{\sqrt {{u_k}{\lambda _k}^2t} }}{{{\Omega _k}}}} \right)}^{2{\ell_k}}}}}} \Psi \left( {m + {\ell_k}, m + {\ell_k} + s;\frac{1}{{{\Omega _k}}}} \right),
\end{align}
where $\Psi \left( \cdot \right)$ denotes Tricomi's confluent hypergeometric function \cite[Eq.9.211.4]{gradshteyn1965table}. Plugging (\ref{eqn_expre_mellin_r_expansion}) into (\ref{eqn_mellin_Rk}) together with some algebraic manipulations yields
\begin{equation}\label{eqn_mellin_Rk1}
\phi \left( {\left. s \right|t} \right)
= \sum\limits_{{\ell_1}, \cdots ,{\ell_K} = 0}^\infty  {{t^{\sum\limits_{k = 1}^K {{\ell_k}} }}{e^{ - t\sum\limits_{k = 1}^K {\frac{{{\lambda _k}^2}}{{1 - {\lambda _k}^2}}} }}\prod\limits_{k = 1}^K {\frac{{{{\left( {\frac{{{\lambda _k}^2}}{{1 - {\lambda _k}^2}}} \right)}^{{\ell_k}}}}}{{{\ell_k}!{{\left( {{\Omega _k}} \right)}^{m + {\ell_k}}}}}\Psi \left( {m + {\ell_k},m + {\ell_k} + s;\frac{1}{{{\Omega _k}}}} \right)} }.
\end{equation}
Meanwhile, by using inverse Mellin transform, the conditional PDF $f_{\left. G_K \right|T}(x|t)$ can be written as \cite{debnath2010integral}
\begin{equation}\label{eqn_pdf_G_on_T}
{f_{\left. G_K \right|T}}\left( {\left. x \right|t} \right) = \left\{ {{\mathcal M^{ - 1}}\phi } \right\}\left( \left. x \right| t\right) = \frac{1}{{2\pi \rm i}}\int\nolimits_{c - {\rm i}\infty }^{c + {\rm i}\infty } {{x^{ - s}}} \phi \left( {\left. s \right|t} \right)ds,
\end{equation}
where ${\rm i} = \sqrt{-1}$. Putting (\ref{eqn_mellin_Rk1}) into (\ref{eqn_pdf_G_on_T}) finally leads to (\ref{eqn:g_hat_cond_t}).



\section{Proof of Theorem \ref{the:cdf_pdf_corr_gam_shif}} \label{app:fox_H_fun}
Based on (\ref{eqn_pdf_f_g_1_p}), the CDF of the product of multiple correlated shifted-Gamma RVs ${F_{ G_K}}\left( x \right)$ directly follows as (\ref{eqn:CDF_G_def_sec}) where the CDF of $\mathcal A_{\bs \ell}$ can be written based on (\ref{eqn:pdf_produc_gamma_ind_sec}) as
\begin{equation}\label{eqn:CDF_A_ell_def}
{F_{{{\cal A}_{\bs \ell} }}}(x) 
 = \frac{1}{{\prod\nolimits_{k = 1}^K {{\Omega _k}} }}\int_0^x {Y_{0,K}^{K,0}\left[ {\left. {\begin{array}{*{20}{c}}
 - \\
{\left( {0,1,\frac{1}{{{\Omega _1}}},m + {\ell _1}} \right), \cdots ,\left( {0,1,\frac{1}{{{\Omega _K}}},m + {\ell _K}} \right)}
\end{array}} \right|\frac{t}{{\prod\nolimits_{k = 1}^K {{\Omega _k}} }}} \right]dt} .
\end{equation}
Regarding to the generalized Fox's H function in (\ref{eqn:CDF_A_ell_def}), we have the following properties.

\begin{property}\label{prop:1}
\begin{multline}
{t^\rho }Y_{p,q}^{m,n}\left[ {\left. {\begin{array}{*{20}{c}}
{\left( {{a_1},{\alpha _1},{A_1},{\varphi _1}} \right), \cdots ,\left( {{a_p},{\alpha _p},{A_p},{\varphi _p}} \right)}\\
{\left( {{b_1},{\beta _1},{B_1},{\phi _1}} \right), \cdots ,\left( {{b_q},{\beta _q},{B_q},{\phi _q}} \right)}
\end{array}} \right|t} \right]\\
 = Y_{p,q}^{m,n}\left[ {\left. {\begin{array}{*{20}{c}}
{\left( {{a_1} + \rho {\alpha _1},{\alpha _1},{A_1},{\varphi _1}} \right), \cdots ,\left( {{a_p} + \rho {\alpha _q},{\alpha _p},{A_p},{\varphi _p}} \right)}\\
{\left( {{b_1} + \rho {\beta _1},{\beta _1},{B_1},{\phi _1}} \right), \cdots ,\left( {{b_q} + \rho {\beta _q},{\beta _q},{B_q},{\phi _q}} \right)}
\end{array}} \right|t} \right].
\end{multline}
\end{property}

\begin{property}\label{prop:2}
\begin{multline}\label{eqn:x_1_fox_H}
Y_{p,q}^{m,n}\left[ {\left. {\begin{array}{*{20}{c}}
{\left( {{a_1},{\alpha _1},{A_1},{\varphi _1}} \right), \cdots ,\left( {{a_p},{\alpha _p},{A_p},{\varphi _p}} \right)}\\
{\left( {{b_1},{\beta _1},{B_1},{\phi _1}} \right), \cdots ,\left( {{b_q},{\beta _q},{B_q},{\phi _q}} \right)}
\end{array}} \right|t} \right]\\
 = Y_{q,p}^{n,m}\left[ {\left. {\begin{array}{*{20}{c}}
{\left( { 1 - {b_1},{\beta _1},{B_1},{\phi _1}} \right), \cdots ,\left( { 1 - {b_q},{\beta _q},{B_q},{\phi _q}} \right)}\\
{\left( { 1 - {a_1},{\alpha _1},{A_1},{\varphi _1}} \right), \cdots ,\left( { 1 - {a_p},{\alpha _p},{A_p},{\varphi _p}} \right)}
\end{array}} \right|t^{-1}} \right].
\end{multline}
\end{property}

Notice that Property \ref{prop:1} is similar to the property of Meijer-G function in \cite[Eq.9.31.5]{gradshteyn1965table} and the property of generalized upper incomplete Fox's H functions in \cite[Eq.A.10]{yilmaz2009productshifted}, while Property \ref{prop:2} is similar to the property of Meijer-G function in \cite[Eq.9.31.2]{gradshteyn1965table} and the property of generalized upper incomplete Fox's H functions in \cite[Eq.A.9]{yilmaz2009productshifted}. They can thus be proved using similar approaches in \cite{yilmaz2009productshifted} and \cite{gradshteyn1965table}.

By using Property \ref{prop:1}, (\ref{eqn:CDF_A_ell_def}) can be derived as
\begin{equation}\label{eqn:F_a_ell_step1}
{F_{{{\cal A}_{\bs \ell} }}}(x) = \int_0^x {{t^{ - 1}}Y_{0,K}^{K,0}\left[ {\left. {\begin{array}{*{20}{c}}
 - \\
{\left( {1,1,\frac{1}{{{\Omega _1}}},m + {\ell _1}} \right), \cdots ,\left( {1,1,\frac{1}{{{\Omega _K}}},m + {\ell _K}} \right)}
\end{array}} \right|\frac{t}{{\prod\nolimits_{k = 1}^K {{\Omega _k}} }}} \right]dt}.
\end{equation}

Then applying Property \ref{prop:2} into (\ref{eqn:F_a_ell_step1}) yields
\begin{equation}\label{eqn:F_a_ell_step2}
{F_{{{\cal A}_{\bs \ell} }}}(x) = \int_0^x {{t^{ - 1}}Y_{K,0}^{0,K}\left[ {\left. {\begin{array}{*{20}{c}}
{\left( {0,1,\frac{1}{{{\Omega _1}}},m + {\ell _1}} \right), \cdots ,\left( {0,1,\frac{1}{{{\Omega _K}}},m + {\ell _K}} \right)}\\
 -
\end{array}} \right|\frac{{\prod\nolimits_{k = 1}^K {{\Omega _k}} }}{t}} \right]dt}.
\end{equation}

With the definition of generalized Fox's H function \cite{yilmaz2010outage,chelli2013performance}, (\ref{eqn:F_a_ell_step2}) can be further derived as
\begin{align}\label{eqn:F_a_ell_step3}
{F_{{{\cal A}_{\bs \ell} }}}(x) &= \frac{1}{{\prod\nolimits_{k = 1}^K {{\Omega _k}} }}\frac{1}{{2\pi {\rm i}}}\int_{c_1-{\rm i}\infty}^{c_1-{\rm i}\infty} {\int_0^x {{{\left( {\frac{t}{{\prod\nolimits_{k = 1}^K {{\Omega _k}} }}} \right)}^{s - 1}}M_{K,0}^{0,K}\left( s \right)dt} ds} \notag\\
 &= \frac{1}{{2\pi {\rm i}}}\int_{c_1-{\rm i}\infty}^{c_1-{\rm i}\infty} {{{\left( {\frac{x}{{\prod\nolimits_{k = 1}^K {{\Omega _k}} }}} \right)}^s}\frac{{\Gamma \left( s \right)}}{{\Gamma \left( {s + 1} \right)}}M_{K,0}^{0,K}\left( s \right)ds},
\end{align}
where $M_{K,0}^{0,K}\left( s \right) = \frac{1}{{\prod\nolimits_{j = 1}^K {\Xi \left( {0,1,\frac{1}{{{\Omega _j}}},m + {\ell _j}} \right)} }}$ 
and $\Xi \left( {a,\alpha ,A,\varphi } \right) = {A^{\varphi  + a + \alpha s - 1}}\Psi \left( {\varphi ,\varphi  + a + \alpha s;A} \right)$.
The function $\Xi \left( {a,\alpha ,A,\varphi } \right)$ can be rewritten by using one important property of Tricomi's confluent hypergeometric function $\Psi \left( {\alpha ,\gamma ;z} \right)$ in \cite[Eq.9.210.2]{gradshteyn1965table} as
\begin{multline}\label{eqn:Xi_rew}
\Xi \left( {a,\alpha ,A,\varphi } \right) = {A^{\varphi  + a + \alpha s - 1}}\frac{{\Gamma \left( {1 - \left( {\varphi  + a + \alpha s} \right)} \right)}}{{\Gamma \left( {1 - a - \alpha s} \right)}}{}_1{F_1}\left( {\varphi ,\varphi  + a + \alpha s;A} \right)\\
 + \frac{{\Gamma \left( {\varphi  + a + \alpha s - 1} \right)}}{{\Gamma \left( \varphi  \right)}}{}_1{F_1}\left( {1 - a - \alpha s,2 - \left( {\varphi  + a + \alpha s} \right);A} \right),
\end{multline}
where ${}_{1}F_{1}(\cdot)$ represents the confluent hypergeometric function \cite[Eq.9.210.1]{gradshteyn1965table}. Thus it follows that
\begin{equation}\label{eqn:gamma__plus_s}
\Xi \left( {a,1,0,1} \right)= \Gamma \left( {a + s} \right) = \mathop {\lim }\limits_{u \to 0} {u^{a + s}}\Psi \left( {1,1 + a + s;u} \right), \, a+s > 0.
\end{equation}

By using (\ref{eqn:gamma__plus_s}) together with the definition of generalized Fox's H function, (\ref{eqn:F_a_ell_step3}) can be further rewritten as
\begin{multline}\label{eqn:use_integral_inv_fox_H_1}
{F_{{{\cal A}_{\bs \ell} }}}(x) = \mathop {\lim }\limits_{u \to 0} \frac{1}{{2\pi {\rm i}}}\int_{c_1-{\rm i}\infty}^{c_1+{\rm i}\infty} {{{\left( {\frac{x}{{\prod\nolimits_{k = 1}^K {{\Omega _k}} }}} \right)}^s}\frac{{{u^s}\Psi \left( {1,1 + s;u} \right)}}{{{u^{1 + s}}\Psi \left( {1,2 + s;u} \right)\prod\nolimits_{j = 1}^K {\Xi \left( {0,1,\frac{1}{{{\Omega _j}}},m + {\ell _j}} \right)} }}ds} \\
 = Y_{K + 1,1}^{1,K}\left[ {\left. {\begin{array}{*{20}{c}}
{\left( {0,1,\frac{1}{{{\Omega _1}}},m + {\ell _1}} \right), \cdots ,\left( {0,1,\frac{1}{{{\Omega _K}}},m + {\ell _K}} \right),\left( {1,1,0,1} \right)}\\
{\left( {0,1,0,1} \right)}
\end{array}} \right|\frac{{\prod\nolimits_{k = 1}^K {{\Omega _k}} }}{x}} \right].
\end{multline}
Finally, applying Property \ref{prop:2} to (\ref{eqn:use_integral_inv_fox_H_1}) yields (\ref{eqn:CDF_F_A_def_sec}).

\section{Proof of (\ref{eqn:trun_err_fina_hyp})}\label{app:uperb_err}
Following from (\ref{eqn:trun_err_prod_gam}), the truncation error $\nabla(N)$ is upper bounded as
\begin{align}\label{eqn:trun_err_upp}
\nabla(N) &= \sum\nolimits_{n = N + 1}^\infty  {\sum\nolimits_{{\sum\nolimits_{k=1}^K  \ell_k = n}}  {{W_{\bs{\ell}}}{F_{{ {\cal A}_{\bs{\ell}}}}}\left( 2^{\cal R} \right)} }\notag\\
 &\le {{F_{{ {\cal A}_{\bs{\ell}}},N}^{\rm max}}\left( 2^{\cal R} \right)} \sum\nolimits_{n = N + 1}^\infty  {\sum\nolimits_{{{\sum\nolimits_{k=1}^K  \ell_k = n}}} {{W_{\bs{\ell}}}} }\notag\\
& 
 =
 {{F_{{ {\cal A}_{\bs{\ell}}},N}^{\rm max}}\left( 2^{\cal R} \right)} {W_{\bf{0}}}\sum\nolimits_{n = N + 1}^\infty  {\frac{{{{\left( m \right)}_n}}}{{n!}}\sum\nolimits_{{\sum\nolimits_{k=1}^K  \ell_k = n}} {n!\prod\nolimits_{k = 1}^K {\frac{{{{\left( {{w_k}} \right)}^{{\ell_k}}}}}{{{\ell_k}!}}} } },
\end{align}
where ${{F_{{ {\cal A}_{\bs{\ell}}},N}^{\rm max}}\left( 2^{\cal R} \right)} = \mathop {\max }\nolimits_{\sum\nolimits_{k=1}^K  \ell_k \ge N + 1} \left( {{F_{{ {\cal A}_{\bs{\ell}}}}}\left( 2^{\cal R} \right)} \right)$ and $(\cdot)_n$ denotes Pochhammer symbol. Herein, it is readily found that $\mathop {\max }\nolimits_{\sum\nolimits_{k=1}^K  \ell_k = N_1} \left( {{F_{{ {\cal A}_{\bs{\ell}}}}}\left( 2^{\cal R} \right)} \right) \ge \mathop {\max }\nolimits_{\sum\nolimits_{k=1}^K  \ell_k = N_2} \left( {{F_{{ {\cal A}_{\bs{\ell}}}}}\left( 2^{\cal R} \right)} \right)$ if $N_1 < N_2$ by the physical interpretation of ${{F_{{ {\cal A}_{\bs{\ell}}}} \left( 2^{\cal R} \right)}}$ in Section \ref{sec:exa}. Thus we have ${{F_{{ {\cal A}_{\bs{\ell}}},N}^{\rm max}}\left( 2^{\cal R} \right)} = \mathop {\max }\nolimits_{\sum\nolimits_{k=1}^K  \ell_k = N + 1} \left( {{F_{{ {\cal A}_{\bs{\ell}}}}}\left( 2^{\cal R} \right)} \right)$. By means of multinomial theorem \cite[Eq.26.4.9]{NIST:DLMF}, it yields
\begin{align}\label{eqn:trun_erro_mul_exp}
\nabla(N)  
 &\le {W_{\bf{0}}}{{F_{{ {\cal A}_{\bs{\ell}}},N}^{\rm max}}\left( 2^{\cal R} \right)} \sum\nolimits_{n = N + 1}^\infty  {\frac{{{{\left( m \right)}_n}}}{{n!}}{{\left( {\sum\nolimits_{k = 1}^K {{w_k}} } \right)}^n}}.
\end{align}
Reformulating the summation term in the right hand side of (\ref{eqn:trun_erro_mul_exp}) in terms of hypergeometric function ${}_2F_1(\cdot)$ [29,Eq.15.1.1], the truncation error $\nabla(N)$ is consequently upper bounded as (\ref{eqn:trun_err_fina_hyp}).
\section{Proof of Theorem \ref{lemma:lemma_gener_fox_H}}\label{app:lemma_gener_fox_H}
Recalling from Theorem \ref{the:cdf_pdf_corr_gam_shif} that ${F_{{{ {\mathcal A}}_{\bs{\ell}}}}}(x)$ is the CDF of the product of multiple independent shifted-Gamma RVs, i.e., $\mathcal A_{\bs{\ell}} = \prod\nolimits_{k = 1}^K (1+R_{{\bs{\ell}},k})$ and $R_{{\bs{\ell}},k} \sim \mathcal {G}(m+\ell_k,{\Omega _k})$, it can be expressed as
\begin{equation}\label{eqn:F_hat_A_1}
{F_{{{ {\mathcal A}}_{\bs{\ell}}}}}(x) = \Pr \left( {\prod\nolimits_{k = 1}^K {\left( {1 + {R_{{\bs{\ell}},k}}} \right)}  \le x} \right) 
=\int\nolimits_{\prod\nolimits_{k = 1}^K {\left( {1 + {t_k}} \right)}  \le x} {\prod\nolimits_{k = 1}^K {\frac{{{t_k}^{m + {\ell_k} - 1}{e^{ - \frac{{{t_k}}}{{{\Omega _k}}}}}}}{{{{\left( {{\Omega _k}} \right)}^{m + {\ell_k}}}\Gamma \left( {m + {\ell_k}} \right)}}} d{t_1} \cdots d{t_{K - 1}}d{t_K}}.
\end{equation}
By using Maclaurin series of exponential function, ${F_{{{ {\mathcal A}}_{\bs{\ell}}}}}(x)$ in (\ref{eqn:F_hat_A_1}) can then be derived as
\begin{align}\label{eqn:F_hat_A_1_fur}
{F_{{{ {\mathcal A}}_{\bs{\ell}}}}}(x)
&= \prod\nolimits_{k = 1}^K {\frac{{\frac{1}{{{{\left( {{\Omega _k}} \right)}^{m + {\ell_k}}}}}}}{{\Gamma \left( {m + {\ell_k}} \right)}}}
 \int\nolimits_{\prod\nolimits_{k = 1}^K {\left( {1 + {t_k}} \right)}  \le x} {\prod\nolimits_{k = 1}^K {{t_k}^{m + {\ell_k} - 1}\sum\nolimits_{{n_1}, \cdots ,{n_K} = 0}^\infty  {\frac{{{{\left( { - \frac{{{t_k}}}{{{\Omega _k}}}} \right)}^{{n_k}}}}}{{{n_k}!}}} } d{t_1} \cdots d{t_{K - 1}}d{t_K}}  \notag\\
 &= \prod\nolimits_{k = 1}^K {\frac{1}{{{{\left( {{\Omega _k}} \right)}^{m + {\ell_k}}}}}} \sum\nolimits_{{n_1}, \cdots ,{n_K} = 0}^\infty  {\prod\nolimits_{k = 1}^K {\frac{1}{{\Gamma \left( {m + {\ell_k}} \right){n_k}!}}{{\left( { - \frac{1}{{{\Omega _k}}}} \right)}^{{n_k}}}} } {g_{{\bf{n}} + {\bs{\ell}}}}\left( x \right),
\end{align}
where ${\bf n} = (n_1,\cdots,n_K)$ and ${g_{{\bs{\ell}}}}\left( x \right)$ is given by (\ref{eqn:sec_def_g_fun}). After putting $\Omega _k = \frac{{P_k}{\sigma _k}^2\left( {1 - {\lambda _k}^2} \right)}{m\mathcal N_0}$ and (\ref{eqn:gamma_snr_rel}) into (\ref{eqn:F_hat_A_1_fur}), the CDF expression in (\ref{eqn:der_A_1_CDF}) directly follows.
On the other hand, ${F_{{{ {\cal A}}_{\bs{\ell}}}}}(x)$ can also be expressed in the form of Mellin-Barnes integral from (\ref{eqn:F_a_ell_step3}) as
\begin{align}\label{eqn:F_A0_brane_inte}
{F_{{{ {\cal A}}_{\bs{\ell}}}}}(x) &= \frac{1}{{2\pi {\rm{i}}}}\int\nolimits_{c_1 - {\rm{i}}\infty }^{c_1 + {\rm{i}}\infty } {\frac{{\Gamma \left( s \right)}}{{\Gamma \left( {s + 1} \right)}}\prod\nolimits_{k = 1}^K {\frac{{\Psi \left( {m ,m + \ell_k + 1 - s;\frac{1}{{{\Omega _k}}}} \right)}}{{{{\left( {{\Omega _k}} \right)}^{m+\ell_k}}}}} {x^s}ds} \notag \\
&=\frac{1}{{2\pi {\rm{i}}}}\int_{{c_2} - {\rm{i}}\infty }^{{c_2} + {\rm{i}}\infty } {\underbrace {\frac{{\Gamma \left( { - s} \right)}}{{\Gamma \left( { - s + 1} \right)}}\prod\nolimits_{k = 1}^K {\frac{{\Psi \left( {m,m + {\ell _k} + 1 + s;\frac{1}{{{\Omega _k}}}} \right)}}{{{{\left( {{\Omega _k}} \right)}^{m + {\ell _k}}}}}} }_{\left\{ {\mathcal M{F_{{{\cal A}_{\bs \ell} }}}} \right\}\left( s \right)}{x^{ - s}}ds}.
\end{align}
where
$c_2 = -c_1$ and the fundamental strip\footnote{$ \left\{ {\mathcal M{{F_{{{ {\cal A}}_{\bs{\ell}}}}}}} \right\}\left( s \right)$ exists for any complex number $s$ in the fundamental strip.} of $ \left\{ {\mathcal M{{F_{{{ {\cal A}}_{\bs{\ell}}}}}}} \right\}\left( s \right)$ implies $c_2 \in (-\infty,0)$ because $ {F_{{\mathcal A_{\bs \ell} }}}\left( x \right)$ admits\cite[p400]{szpankowski2010average}
 \begin{equation}\label{eqn:fundamental_steip_F_al}
 {F_{{\mathcal A_{\bs \ell} }}}\left( x \right) = \left\{ {\begin{array}{*{20}{c}}
{\mathcal O\left( {{x^\infty }} \right)}&{x \to 0}\\
{\mathcal O\left( {{x^0}} \right)}&{x \to \infty }
\end{array}} \right.
 \end{equation}
 where $\mathcal O(\cdot)$ denotes the big O notation and the first equation of (\ref{eqn:fundamental_steip_F_al}) holds because of ${F_{{\mathcal A_{\bs \ell} }}}\left( x \right) = 0$ for $x \le 1$. Thus $c_1 \in (0,\infty)$.

By adopting \cite[Eq.9.210.2]{gradshteyn1965table} into (\ref{eqn:F_A0_brane_inte}), it yields
\begin{multline}\label{eqn:F_hat_A_0_der_fur_ex}
{F_{{{ {\cal A}}_{\bs{\ell}}}}}(x) = \frac{1}{{2\pi {\rm{i}}}}\int\nolimits_{c_1 - {\rm{i}}\infty }^{c_1 + {\rm{i}}\infty } {\frac{{\Gamma \left( s \right)}}{{\Gamma \left( {s + 1} \right)}} } \\
\times \prod\limits_{k = 1}^K {\frac{1}{{{{\left( {{\Omega _k}} \right)}^{m+\ell_k }}}}}\left( {\begin{array}{*{20}{l}}
{\frac{{\Gamma \left( { - m -\ell_k+ s} \right)}}{{\Gamma \left( s \right)}}{}_1{F_1}\left( {m+\ell_k ,m+\ell_k + 1 - s;\frac{1}{{{\Omega _k}}}} \right) + }\\
{\frac{{\Gamma \left( {m+\ell_k - s} \right)}}{{\Gamma \left( {m+\ell_k } \right)}}{{\left( {\frac{1}{{{\Omega _k}}}} \right)}^{s - m-\ell_k }}{}_1{F_1}\left( {s,1 - m-\ell_k + s;\frac{1}{{{\Omega _k}}}} \right)}
\end{array}} \right){x^s}ds.
\end{multline}
where ${}_1{F_1}\left( {\alpha ,\beta ;\frac{1}{{{\Omega _k}}}} \right)$ can be expanded as
\begin{equation}\label{eqn:F_11_series_exp}
{}_1{F_1}\left( {\alpha ,\beta ;\frac{1}{{{\Omega _k}}}} \right) = 1 + \frac{\beta }{\alpha }\frac{1}{{{\Omega _k}}} + o\left( {\frac{1}{{{\Omega _k}}}} \right).
\end{equation}
Clearly, as $\gamma$ approaches to infinity, the dominant term in (\ref{eqn:F_11_series_exp}) is $1$.

Herein, we assume $c_1>m+{\rm max}\{\bs{\ell}\}$ because $c_1 $ could be any point in $(0,\infty)$. Then by substituting (\ref{eqn:F_11_series_exp}) into (\ref{eqn:F_hat_A_0_der_fur_ex}), we have
\begin{align}\label{eqn:F_A_0_asym_F_hy_exp}
&{{F_{{{ {\cal A}}_{\bs{\ell}}}}}(x)}{ = \prod\limits_{k = 1}^K {\frac{1}{{{{\left( {{\Omega _k}} \right)}^{m+\ell_k}}}}} \frac{1}{{2\pi {\rm{i}}}}\int\nolimits_{{c_1} - {\rm{i}}\infty }^{{c_1} + {\rm{i}}\infty } {\frac{{\Gamma \left( s \right)}}{{\Gamma \left( {s + 1} \right)}}\prod\limits_{k = 1}^K {\frac{{\Gamma \left( { - m-\ell_k + s} \right)}}{{\Gamma \left( s \right)}}} {x^s}ds}  + o\left( {\prod\limits_{k = 1}^K {\frac{1}{{{{\left( {{\Omega _k}} \right)}^{m+\ell_k}}}}} } \right)} \notag\\
{}&{ = \prod\limits_{k = 1}^K {\frac{1}{{{{\left( {{\Omega _k}} \right)}^{m +\ell_k}}}}} G_{K + 1,K + 1}^{0,K + 1}\left( {\left. {\begin{array}{*{20}{c}}
{1,1 +\ell_1+ m , \cdots ,1 +\ell_K+ m}\\
{1, \cdots ,1,0}
\end{array}} \right|x} \right) + o\left( {\prod\limits_{k = 1}^K {\frac{1}{{{{\left( {{\Omega _k}} \right)}^{m+\ell_k }}}}} } \right)}.
\end{align}
Putting (\ref{eqn:gamma_snr_rel}) into (\ref{eqn:F_A_0_asym_F_hy_exp}), ${{F_{{{ {\cal A}}_{\bs{\ell}}}}}(x)}$ can be finally written as
\begin{align}\label{eqn:F_A_0_asym_F_hy_exp1}
{{F_{{{ {\cal A}}_{\bs{\ell}}}}}(x)}&{ = \prod\limits_{k = 1}^K {{{\left( {\frac{m}{{{\theta _k}{\sigma _k}^2\left( {1 - {\lambda _k}^2} \right)}}} \right)}^{m +\ell_k }}} G_{K + 1,K + 1}^{0,K + 1}\left( {\left. {\begin{array}{*{20}{c}}
{1,1 +\ell_1+ m, \cdots ,1 +\ell_K+m}\\
{1, \cdots ,1,0}
\end{array}} \right|x} \right){\gamma ^{-d_{{ {\cal A}_{\bs{\ell}}}}}} }\notag \\
&+ o\left( {{\gamma ^{-d_{{ {\cal A}_{\bs{\ell}}}}}}} \right).
\end{align}
Comparing (\ref{eqn:F_A_0_asym_F_hy_exp1}) with (\ref{eqn:der_A_1_CDF}) and setting the coefficients corresponding to ${\gamma ^{-d_{{ {\cal A}_{\bs{\ell}}}}}}$ equal, we finally have (\ref{eqn:g_0_0_der_meij_rem}).

\section{Proof of Lemma \ref{cor:time_corr}} \label{app:proo_cor}
Defining ${\bs{\lambda }_1} = \left( {{\lambda _1}, \cdots,{\lambda _K}} \right)$ and given $K>1$, we first consider the simplest case where $\bs{\lambda}_1 \preceq \bs{\lambda}_2$ and ${\bs{\lambda }_2} = \left( {{\lambda _1}+ {\Delta _1}, \cdots ,{\lambda _K}} \right)$ with ${\Delta _1} \ge 0$. With (\ref{eqn:ell_def}), $\ell \left( {{\bs \lambda _2},K} \right)$ can be written as
\begin{equation}\label{eqn:ell_lambda_2_def}
\ell \left( {{\bs \lambda _2},K} \right) = \left( {1 + \sum\limits_{k = 2}^K {\frac{{{\lambda _k}^2}}{{1 - {\lambda _k}^2}}}  + \frac{{{{\left( {{\lambda _1} + {\Delta _1}} \right)}^2}}}{{1 - {{\left( {{\lambda _1} + {\Delta _1}} \right)}^2}}}} \right)\left( {1 - {{\left( {{\lambda _1} + {\Delta _1}} \right)}^2}} \right)\prod\limits_{k = 2}^K {\left( {1 - {\lambda _k}^2} \right)}.
\end{equation}
Since ${\Delta _1} \ge 0$, we have
\begin{align}\label{eqn:lambda_2_writ_re}
\ell \left( {{\bs \lambda _2},K} \right) &= \left( {1 + \sum\limits_{k = 2}^K {\frac{{{\lambda _k}^2\left( {1 - {{\left( {{\lambda _1} + {\Delta _1}} \right)}^2}} \right)}}{{1 - {\lambda _k}^2}}} } \right)\prod\limits_{k = 2}^K {\left( {1 - {\lambda _k}^2} \right)} \notag\\
 &\le \left( {1 + \sum\limits_{k = 2}^K {\frac{{{\lambda _k}^2\left( {1 - {\lambda _1}^2} \right)}}{{1 - {\lambda _k}^2}}} } \right)\prod\limits_{k = 2}^K {\left( {1 - {\lambda _k}^2} \right)}
= \ell \left( {{\bs \lambda _1},K} \right),
\end{align}
where the equality holds if and only if ${\Delta _1} = 0$. Now gradually adding non-negative increments $\Delta_i$ to the other elements $\lambda_i$, $i=2, \cdots, K$, in ${\bs \lambda _2}$ and using the similar approach as (\ref{eqn:lambda_2_writ_re}), we can finally have $\ell \left( {{\bs \lambda _2},K} \right) \le \ell \left( {{\bs \lambda _1},K} \right)$ for general ${\bs{\lambda }_2} = \left( {{\lambda _1}+ {\Delta _1}, \cdots, {\lambda _i}+{\Delta _i}, \cdots,{\lambda _K}+{\Delta _K}} \right)$ with ${\Delta _i} \ge 0$ where the equality holds if and only if ${\Delta _i} =0$. Clearly, when $K > 1$, we have $\ell \left( {{\bs \lambda _1},K} \right) \le \ell \left( {{\bf 0},K} \right)=1$, where the equality holds if and only if ${\bs \lambda _1} = {\bf 0}$. Then the proof completes.

\section{Proof of Lemma \ref{the:coding_modulation_gain}}\label{app:coding_gain}

It is readily found from (\ref{eqn:sec_def_g_fun}) that ${g_{\bs{\ell}}}\left( 2^{\mathcal R} \right)$ is a monotonically increasing function of $\mathcal R$. Therefore the first derivative of ${g_{\bs{\ell}}}\left( 2^{\mathcal R} \right)$ with respect to $\mathcal R$ is greater than $0$, i.e., ${g_{\bs{\ell}}}^\prime \left( 2^{\mathcal R} \right) > 0$. By applying the property of derivatives of Laplace transform into (\ref{eqn:g_0_0_der_meij_rem}) \cite{debnath2010integral}, it follows that
\begin{align}\label{eqn:first_der_g}
{g_{\bs{\ell}}}^\prime \left( {{2^{\cal R}}} \right) &= \Theta \prod\limits_{k = 1}^K {\Gamma \left( {m + {\ell_k}} \right)} \frac{1}{{2\pi {\rm{i}}}}\int_{{c_1} - {\rm{i}}\infty }^{{c_1} + {\rm{i}}\infty } {\prod\limits_{k = 1}^K {\frac{{\Gamma \left( { - m - {\ell_k} + s} \right)}}{{\Gamma \left( s \right)}}} {2^{{\cal R}s}}ds} > 0. 
\end{align}
where $\Theta = \ln 2$. Similarly, the second derivative of ${g_{\bs{\ell}}}\left( 2^{\mathcal R} \right)$ with respect to $\mathcal R$ is given by
\begin{equation}\label{eqn:sec_der_g}
{g_{\bs{\ell}}}^{\prime \prime }\left( 2^{\mathcal R} \right) = {\Theta^2}\prod\limits_{k = 1}^K {\Gamma \left( {m + {\ell_k}} \right)} \frac{1}{{2\pi {\rm{i}}}}\int_{{c_1} - {\rm{i}}\infty }^{{c_1} + {\rm{i}}\infty } {s\prod\limits_{k = 1}^K {\frac{{\Gamma \left( { - m - {\ell_k} + s} \right)}}{{\Gamma \left( s \right)}}} {2^{{\cal R}s}}ds}.
\end{equation}
Without loss of generality, we assume that $\ell_1={\rm max}\{\bs{\ell}\}$. By rewriting $s=(s - m - \ell_1) + (m + \ell_1)$, the integral in (\ref{eqn:sec_der_g}) can be derived as
\begin{multline}\label{eqn:mellin_g_0_secd}
\frac{1}{{2\pi {\rm{i}}}}\int_{{c_1} - {\rm{i}}\infty }^{{c_1} + {\rm{i}}\infty } {s\prod\limits_{k = 1}^K {\frac{{\Gamma \left( { - m - {\ell_k} + s} \right)}}{{\Gamma \left( s \right)}}} {2^{\mathcal Rs}}ds}
=  \left( {m + {\ell_1}} \right)\frac{1}{{2\pi {\rm{i}}}}\int_{{c_1} - {\rm{i}}\infty }^{{c_1} + {\rm{i}}\infty } {\prod\limits_{k = 1}^K {\frac{{\Gamma \left( { - m - {\ell_k} + s} \right)}}{{\Gamma \left( s \right)}}} {2^{\mathcal Rs}}ds}\\
 + \frac{1}{{2\pi {\rm{i}}}}\int_{{c_1} - {\rm{i}}\infty }^{{c_1} + {\rm{i}}\infty } {\frac{{\Gamma \left( {1 - m - {\ell_1} + s} \right)}}{{\Gamma \left( s \right)}}\prod\limits_{k = 2}^K {\frac{{\Gamma \left( { - m - {\ell_k} + s} \right)}}{{\Gamma \left( s \right)}}} {2^{\mathcal Rs}}ds}.
\end{multline}
Putting (\ref{eqn:mellin_g_0_secd}) into (\ref{eqn:sec_der_g}) along with (\ref{eqn:first_der_g}), it follows that
\begin{equation}\label{eqn:sec_der_g_fin}
{g_{\bs{\ell}}}^{\prime \prime }\left( x \right){\rm{ = }}\left\{ {\begin{array}{*{20}{c}}
{\Theta {g_{{{\bs{\ell}}_{/1}}}}^\prime \left( {{2^{\cal R}}} \right) + \Theta {g_{\bs{\ell}}}^\prime \left( {{2^{\cal R}}} \right)>0,}&{{\ell_1} + m = 1{\rm{;}}}\\
{\Theta \frac{{\Gamma \left( {m + {\ell_1}} \right)}}{{\Gamma \left( {m + {\ell_1} - 1} \right)}}{g_{\bar {{\bs{\ell}}}}}^\prime \left( {{2^{\cal R}}} \right) + \Theta \left( {m + {\ell_1}} \right){g_{\bs{\ell}}}^\prime \left( {{2^{\cal R}}} \right)>0,}&{{\ell_1} + m > 1.}
\end{array}} \right.
\end{equation}
where ${{\bs{\ell}}_{/1}} = \left( {{\ell_2},{\ell_3}, \cdots ,{\ell_K}} \right)$ and $\bar {{\bs{\ell}}} = \left( {{\ell_1} - 1,{\ell_2},{\ell_3}, \cdots ,{\ell_K}} \right)$. Thus it proves that ${{g_{\bs{\ell}}}\left(  2^{\cal R} \right)}$ is a convex function of $\mathcal R$ if ${\rm max}{\{\bs{\ell}\}}+m \ge 1$.  Then the proof directly follows.

\bibliographystyle{ieeetran}
\bibliography{Asy_ana}

\end{document}